\documentclass[12pt]{article}

\usepackage{amsthm}
\usepackage{amsmath}
\usepackage{amssymb}
\usepackage{amsfonts}
\usepackage{amscd}
\usepackage{graphicx}

\newtheorem{theorem}{Theorem}
 
\newtheorem{corollary}[theorem]{Corollary}

\newtheorem{remark}{Remark}%

\newcommand{\Real}{\mathbb{R}}

\newcommand{\Int}{\mathbb{Z}}

\newcommand{\mf}{\mathfrak}
\newcommand{\mc}{\mathcal}

\begin{document}

\title{Relative Equilibria and Periodic Orbits in a Binary Asteroid Model}

\author{Lennard F. Bakker, Nicholas J. Freeman \\ Brigham Young University, USA} 

\maketitle

\begin{abstract}We present a planar four-body model, called the Binary Asteroid Problem, for the motion of two asteroids (having small but positive masses) moving under the gravitational attraction of each other, and under the gravitational attraction of two primaries (with masses much larger than the two asteroids) moving in uniform circular motion about their center of mass. We show the Binary Asteroid Model has (at least) 6 relative equilibria and (at least) 10 one-parameter families of periodic orbits, two of which are of Hill-type. The existence of six relative equilibria and 8 one-parameter families of periodic orbits is obtained by a reduction of the Binary Asteroid Problem in which the primaries have equal mass, the asteroids have equal mass, and the positions of the asteroids are symmetric with respect to the origin. The remaining two one-parameter families of periodic orbits, which are of comet-type, are obtained directly in the Binary Asteroid Problem.\end{abstract}

\section{Introduction}\label{sec1}
In December 2017, the asteroid 2017 ${\rm YE}_5$ was discovered by Claudine Rinner as part of the Morocco Ouka\"imeden Sky Survey directed by Zouhair Benkhaldoun \cite{Davis}. Only later in June 2018, when 2017 ${\rm YE}_5$ passed close enough to Earth was it determined by ground-based radar imaging that it was a binary asteroid where the two asteroids have nearly equal mass, i.e., the mass ratio of the two asteroids is close to $1$ \cite{nasa}. There are many known binary asteroids with a larger mass central asteroid and a smaller mass satellite asteroid \cite{Johnston}, but only four known near-Earth binary asteroids with nearly equal mass, with 2017 ${\rm YE}_5$ being the fourth discovered. The other three known near-Earth binary asteroids with nearly equal mass are (69230) Hermes, (190166) 2005 ${\rm UP}_{156}$, and 1994 ${\rm CJ}_1$ \cite{Monteiro}. In the main asteroid belt there are two more binary asteroids with nearly equal mass, and these are 90 Antiope and (300163) 2006 ${\rm VW}_{139}$ \cite{Johnston}. Of these six binary asteroids with nearly equal mass, (300163) 2006 ${\rm VW}_{163}$ is also classified as a comet with designation 288P \cite{Johnston}, and 2017 ${\rm YE}_5$ is also considered a possible dormant Jupiter-family comet \cite{Monteiro}.

The standard model for the motion of an asteroid (be it a single asteroid, a binary asteroid with arbitrary mass ratio, etc.) is with six orbital elements (the orbital elements of known binary asteroids are listed in \cite{Johnston}). This model uses observational data of the asteroid and the integrable two-body problem to estimate the asteroid's orbital elements (e.g., see \cite{Espitia}) that describe the ellipse the asteroid approximately makes around the Sun.

Another model for the motion of a binary asteroid is the Restricted Hill Full $4$-Body Problem \cite{Scheeres}, which is known as the Binary Asteroid System. This model has one body with large mass (the Sun) at a large distance from the two bodies with smaller masses (the binary asteroid with arbitrary mass ratio) and a spacecraft (the zero-mass particle). It is assumed that two asteroids and the spacecraft remain close together. This model further assumes each of the asteroids has positive volume and a mass distribution. One of the asteroids is typically assumed to be a homogeneous sphere, and the other is typically assumed to be a constant density triaxial ellipsoid. A recent extension of the Binary Asteroid System model includes the Solar Radiation Pressure in the model, which, using an analytical solution as a seed for a multiple shooting method, has been applied to design a quasiperiodic forced hovering orbit above the one of asteroids for the unequal mass binary asteroid (66391) Moshup and Squannit \cite{Wang}. When both asteroids have irregular shape and heterogeneous mass distributions, a nested interpolation method has been developed, based on a hybrid gravity model on one asteroid and a finite element model of the other, which has been applied to the unequal mass binary asteroid (66391) Moshup and Squannit \cite{Lu}.

In this paper we present a four point-mass model for the motion of a binary asteroid. The spatial version of this model was first proposed for the binary asteroid 2017 ${\rm YE}_5$ \cite{Bakker}, but it applies to any two asteroids. The planar model, formulated here, has two bodies having large masses (called the primaries) and two bodies having small but {\it positive} masses (called the asteroids). We do not assume a priori that the masses of the primaries are equal nor do we assume a priori that the masses of the asteroids are equal. The primaries are prescribed to move in uniform circular motion around their center of mass (set at the origin), as is done in the Circular Restricted Planar Three-Body Problem. The two asteroids move under the gravitational attraction of the two primaries as well as the gravitational attraction of each other. We do not assume that the two asteroids are necessarily close to each other, that they are not a priori in the formation of a binary asteroid. The system of ODEs associated to the Hamiltonian in rotating coordinates (which Hamiltonian is derived in Section 2.2) we call the Binary Asteroid Problem. The analytic existence of 6 relative equilibria and 10 one-parameter families of periodic orbits in the Binary Asteroid Problem presented here form part of the second author's senior thesis \cite{Freeman}.

The Binary Asteroid Problem serves as a ``bridge'' model between two uncoupled copies of the Circular Restricted Planar Three-Body Problem and the general Planar Four-Body Problem. The connection with the Circular Restricted Planar Three-Body Problem will be demonstrated in this paper by relative equilibria of the COM reduced Binary Asteroid Problem obtained by passing to center of mass (COM) coordinates, where the center of mass is with respect to the two asteroids. This COM reduction requires that the masses of the primaries to be equal, that the masses of the asteroids to be equal, and that the two asteroids to be symmetric with respect to each other through their center of mass which set at the origin (Section \ref{COMRed}). As the common mass of the asteroids goes to $0$ each relative equilibria of the COM reduced Binary Asteroid Problem converges to one of the five equilibria of the Generalized Copenhagen Problem (Theorem \ref{ExistenceCritPts}). In a forth-coming paper \cite{Bakker}, we show that the spatial Binary Asteroid Problem inherits the 12 collinear relative equilibria from the general Four-Body Problem, and show that each of these relative equilibria limits to an equilibrium of the Generalized Copenhagen Problem as the masses of the two asteroids (not assumed equal) go to zero. These kinds of connections have been observed between relative equilibria in the equal mass Planar Four-Body Problem \cite{Simo} and the Copenhagen Problem when the value of the mass for a pair of bodies with equal masses goes to zero \cite{Roy}.

The COM reduced Binary Asteroid Problem has an interpretation as a restricted four-body problem. The two primaries and a fictitious third body, placed at the origin, whose mass is that of the common mass of the two asteroids of equal mass, form a collinear central configuration. A zero-mass particle then moves in the gravitational field determined by the three collinear bodies that are moving in uniform circular motion about the origin. In rotating coordinates, the Hamiltonian of the COM reduced Binary Asteroid Problem is similar to that of the Circular Restricted Four-Body Problem with Three Equal Primaries in a Collinear Central Configuration in \cite{Llibre2021}. Since the COM reduced Binary Asteroid Problem has an interpretation as a restricted four-body problem, we will apply and/or adapt techniques used in restricted problems (e.g. \cite{Llibre2021,AlvarezRamirez2015,Stoica,meyerOffin}) in our analysis.

Our motivations for developing the Binary Asteroid Problem consisted of following questions.
\begin{enumerate}
    \item[(1)] If two single asteroids (not in a binary asteroid formation) that are far away from the primaries eventually pass simultaneously close enough to one of the primaries, could the two asteroids be ejected from the close encounter as a binary asteroid?
    \item[(2)] If a binary asteroid far away from the primaries eventually passes close enough to one of the primaries, could the binary asteroid be pulled apart into two single asteroids going in separate directions? 
    \item[(3)] If one asteroid is orbiting a primary and other asteroid that is far away eventually passes close enough to that primary and its orbiting asteroid, could the two asteroids be ejected as a binary asteroid?
    \item[(4)] If a binary asteroid far away from the primaries eventually passes close enough to one of the primaries, could the binary asteroid be pulled apart where one of the asteroid is captured into an orbit around that primary and the other asteroid is ejected?
    \item[(5)] Could a relative equilibrium be the alpha and/or omega limit for the two asteroids, for which the two asteroids are close together for a long time but are pulled apart as they approach the relative equilibria in backward or forward time?
    \item[(6)] Do there exist periodic solutions in a model for binary asteroids in which one of the asteroids orbits one primary and the other asteroid orbits the other primary (a Hill-type orbit)?
\end{enumerate}
We leave questions (1), (2), (3), and (4) for future research, but based on numerical simulations we have realizing affirmative answers for these four questions in the Binary Asteroid Problem (see e.g. \cite{Freeman}), we believe these can be answered analytically through regularization. A partial answer for question (5) for the COM reduced Binary Asteroid Problem is that all the relative equilibria have stable and unstable manifolds and so provide alpha and omega limits for the binary asteroid (Section 3). In \cite{Bakker} we will show that the $12$ collinear relative equilibria in the Binary Asteroid Problem have stable and unstable manifolds. What we do not know, for both the Binary Asteroid Problem and the COM reduced Binary Asteroid Problem, is where those stable and unstable manifolds go in phase space and whether they intersect. An answer for question (6) in the COM reduced Binary Asteroid Problem is that there exists two families of periodic orbits in which one asteroid orbits one primary and the other asteroid orbits the other primary (Section 5). 

The Binary Asteroid Problem readily extends to $3$ or more primaries and $3$ or more asteroids. The $3$ or more primaries, with large masses, set in a central configuration, are prescribed a uniform circular motion about their center of mass, set at the origin. The $3$ or more asteroids, with small but positive small masses relative to the masses of the primaries, move either in a spatial or planar setting in relation to the primaries. Some research with $3$ primaries and $2$ asteroids has been done where the three primaries are in an equilateral triangle with two of the masses of the primaries relatively small with respect to the the remaining primary so that the barycenter of the three primaries is close to the dominant mass primary. In this situation we numerically found relative equilibria on each of the lines that pass through the barycenter of the primaries and the primaries \cite{Cochran}.

This paper is structured as follows. In Section 2 we derive a time-dependent Hamiltonian for the Binary Asteroid Model in a fixed coordinate system. Passing to rotating coordinates, we obtain a time-independent Hamiltonian for the Binary Asteroid Problem. Using COM coordinates we derive the Hamiltonian for the COM reduced Binary Asteroid Problem. In Section 3 we use an amended potential associated to the Hamiltonian of the COM reduced Binary Asteroid Problem, and two symmetries of the amended potential, to prove the existence of six relative equilibria, four of which are collinear with the two primaries, and two of which form convex kite configurations (and are each a rhombus) with the two primaries and have two axes of symmetries. We show how each of the these six relative equilibria limits to the corresponding relative equilibrium in the Generalized Copenhagen Problem. We also show that the four collinear relative equilibria are saddle-centers and that the two kites are hyperbolic. In Section 4 we prove that in the COM reduced Binary Asteroid Problem there exist two one-parameter families of periodic orbits near the origin (the point of collision of the two asteroids). The two asteroids on these periods orbits are close to each other and so form a binary asteroid. In Section 5 we prove in the COM reduced Binary Asteroid Problem the existence of two one-parameter families of near-circular Hill-type orbits in which one asteroid orbits one primary and the other asteroid orbits the other primary. We show numerically that these families continue when the masses of the two primaries vary slightly from being equal. In Section 6 we prove the existence of two one-parameter families of near-circular comet-type orbits in the Binary Asteroid Problem where the angle between the position vectors for the two asteroids is approximately $\pi$. In Section 7 we state our conclusions and future research.

\section{A Binary Asteroid Model}
\subsection{Initial Hamiltonian}
For two bodies with masses $M_1, M_2 > 0$ (hereafter called primaries) having respective positions $\mf{X}_1 = (X_1, Y_1)$, $\mf{X}_2 = (X_2, Y_2)$ in the plane, Newton's law of gravitation says their motion is governed by the equations
\begin{align}
\ddot{\mf{X}}_1 = -\frac{M_2(\mf{X}_1 - \mf{X}_2)}{\|\mf{X}_1 - \mf{X}_2\|^3},&  & \ddot{\mf{X}}_2 = -\frac{M_1(\mf{X}_2 - \mf{X}_1)}{\|\mf{X}_1 - \mf{X}_2\|^3}.
\label{eqn1}
\end{align}
The time unit is chosen so that the gravitation constant is 1. The well-known circular solution to \eqref{eqn1} is given by
\begin{align}
X_1(t) = -\frac{aM_2}{M}\cos\left(\frac{2\pi t}{P}\right), & & Y_1(t) = -\frac{aM_2}{M}\sin\left(\frac{2\pi t}{P}\right), \nonumber\\
X_2(t) = \frac{aM_1}{M}\cos\left(\frac{2\pi t}{P}\right), & & Y_2(t) = \frac{aM_1}{M}\sin\left(\frac{2\pi t}{P}\right),
\label{eqn2}
\end{align}
where $a > 0$ is the distance between the two primaries, $M = M_1 + M_2$ is the total mass of the primaries, and $P > 0$ is the period of the primaries' circular orbit. The three parameters $a$, $M$, and $P$ are related by Kepler's Third Law,
\begin{align}
    \frac{P^2}{a^3} = \frac{4\pi^2}{M}.
    \label{eqn3}
\end{align}
Throughout this paper, we allow $a$ and $M$ to be arbitrary but fixed, so that $P$ is always determined by \eqref{eqn3}.

We introduce two more bodies with masses $m_1, m_2 > 0$ (hereafter called asteroids) at respective positions $q_1 = (x_1, y_1)$, $q_2 = (x_2, y_2)$. These asteroids have kinetic energies
\begin{align*}
K_1 = \frac{p_1^2}{2m_1} = \frac{1}{2m_1}(p_{x_1}^2 + p_{y_1}^2), & & K_2 = \frac{p_2^2}{2m_2} = \frac{1}{2m_2}(p_{x_2}^2 + p_{y_1}^2),
\end{align*}
where $p_1 = (p_{x_1}, p_{y_1})$ and $p_2 = (p_{x_2}, p_{y_2})$ are the corresponding momenta. We assume that
\[ \max\{m_1, m_2\} \ll \min\{M_1,M_2\},\] 
so the asteroids do not appreciably affect the motion of the primaries but do affect each other. Therefore, to a reasonable approximation, the potential energy of asteroids is given by
\begin{align*}
U & = -\frac{m_1m_2}{\|q_1 - q_2\|} - \frac{m_1M_1}{\|\mf{X}_1(t) - q_1\|} - \frac{m_2M_1}{\|\mf{X}_1(t) - q_2\|} \\ 
& \ \ \ \ - \frac{m_1M_2}{\|\mf{X}_2(t) - q_1\|} - \frac{m_2M_2}{\|\mf{X}_2(t) - q_2\|}.
\end{align*}
The time-dependent Hamiltonian for this binary asteroid model is
\begin{align}
H(q_i,p_i,t) & = K_1 + K_2 + U.
\label{eqn4}
\end{align}

\subsection{Rotating Coordinates}
To eliminate the time dependence in \eqref{eqn4}, we use the rotation matrix
\begin{equation*}
\exp\left(\frac{2\pi tK}{P}\right) = \begin{bmatrix}\cos(2\pi t/P) & \sin(2\pi t/P) \\ -\sin(2\pi t/P) & \cos(2\pi t/P)\end{bmatrix}, \text{ where } K = \begin{bmatrix} 0 & 1 \\ -1 & 0 \end{bmatrix}.
\end{equation*}
We make the linear symplectic transformation
\begin{align*}
q_i \mapsto \exp\left(\frac{2\pi tK}{P}\right)q_i, & & p_i \mapsto \exp\left(\frac{2\pi tK}{P}\right)p_i.
\end{align*}
The associated remainder in the rotating coordinates is
\begin{equation*}
R(q_i,p_i) = -\frac{2\pi}{P}\left(q_1^TKp_1 + q_2^TKp_2\right).
\end{equation*}
Since $\exp\left(2\pi tK/P\right)$ is an orthogonal matrix, the norms in the denominators of the potential term are unchanged when we apply the rotation to the distance vectors. The Hamiltonian in the new coordinates is now the time-independent
\begin{align}
H_{BA} 
& = \frac{p_1^2}{2m_1} + \frac{p_2^2}{2m_2} - \frac{2\pi}{P}\left(q_1^TKp_1 + q_2^TKp_2\right) - \frac{m_1m_2}{\|q_1 - q_2\|} \nonumber\\
& \ \ \ \  - \frac{m_1M_1}{\|q_1 + (aM_2/M, 0)\|} - \frac{m_2M_1}{\|q_2 + (aM_2/M, 0)\|} \nonumber\\
& \ \ \ \  - \frac{m_1M_2}{\|q_1 - (aM_1/M, 0)\|} - \frac{m_2M_2}{\|q_2 - (aM_1/M, 0)\|}.
\label{eqn5}
\end{align}
In terms of the components of the positions and momenta, \eqref{eqn5} becomes
\begin{align}
H_{BA} & = \frac{1}{2m_1}\left(p_{x_1}^2 + p_{y_1}^2\right) + \frac{1}{2m_2}\left(p_{x_2}^2 + p_{y_2}^2\right) - \frac{2\pi}{P}\left(x_1p_{y_1} - y_1p_{x_1}\right)\nonumber\\
& \ \ \ \ - \frac{2\pi}{P}\left(x_2p_{y_2} - y_2p_{x_2}\right) - \frac{m_1m_2}{\sqrt{(x_1 - x_2)^2 + (y_1 - y_2)^2}} \nonumber\\
& \ \ \ \ - \frac{m_1M_1}{\sqrt{(x_1 + aM_2/M)^2 + y_1^2}} - \frac{m_2M_1}{\sqrt{(x_2 + aM_2/M)^2 + y_2^2}} \nonumber\\
& \ \ \ \ - \frac{m_1M_2}{\sqrt{(x_1 - aM_1/M)^2 + y_1^2}} - \frac{m_2M_2}{\sqrt{(x_2 - aM_1/M)^2 + y_2^2}}.
\label{eqn6}
\end{align}
The Hamiltonian $H_{BA}$ in \eqref{eqn6} is called the (planar) Binary Asteroid Hamiltonian and the associated system of ODEs is called the (planar) Binary Asteroid Problem.

\subsection{COM Coordinates}\label{COMRed}
We will reduce the number of degrees of freedom in \eqref{eqn6} from $4$ to $2$ by passing to symmetric configurations for the binary asteroids under the assumption of equal masses
\[ m_1 = m_2 = m\]
for the binary asteroids, and under the assumption of equal masses
\[ M_1 = M_2 =  \frac{M}{2}\]
for the primaries. The quantity
\[ \nu = \frac{M_1 - M_2}{2}\] 
measures the deviation between the masses of the primaries and will be use as a perturbation parameter for an expansion of the Hamiltonian. With the matrix 
\begin{equation*}
    \mc{M} = \frac{1}{2}\begin{bmatrix}1 & 0 & -1 & 0\\ 0 & 1 & 0 & -1 \\ 1 & 0 & 1 & 0 \\ 0 & 1 & 0 & 1\end{bmatrix}
\end{equation*}
we define the linear symplectic change of coordinates
\begin{align*}
(x_1,y_1,x_2,y_2)  \mapsto (x,y,\alpha,\beta),  & &  (p_{x_1},p_{y_1},p_{x_2},p_{y_2})  \mapsto (p_x, p_y, p_\alpha, p_\beta),
\end{align*}
by
\begin{align*}
(x,y,\alpha,\beta)^T & = \mc{M}(x_1,y_1,x_2,y_2)^T, \\  (p_x, p_y, p_\alpha, p_\beta)^T & = (\mc{M}^T)^{-1}(p_{x_1},p_{y_1},p_{x_2},p_{y_2})^T.
\end{align*}
Explicitly, the change of coordinates is
\begin{equation} \label{eqn7}
\begin{aligned}
& x = \frac{x_1 - x_2}{2},\ y = \frac{y_1 - y_2}{2},\ \alpha = \frac{x_1 + x_2}{2},\ \beta = \frac{y_1 + y_2}{2}, \\
& p_x = p_{x_1} - p_{x_2},\ p_y = p_{y_1} - p_{y_2},\ p_\alpha = p_{x_1} + p_{x_2},\ p_\beta = p_{y_1} + p_{y_2}.
\end{aligned}
\end{equation}
We call the new coordinates the COM \textit{coordinates} (where COM means ``Center Of Mass"). Physically, $\alpha$ and $\beta$ are the components of the asteroids' barycenter, while $x$ and $y$ are the components of half of the position of the asteroid with mass $m_1$ relative to asteroid with mass $m_2$. The change of coordinates in \eqref{eqn7} is similar to Jacobi coordinates for the two-body problem.

In the COM coordinates the Hamiltonian $H_{BA}$ is analytic in the parameter $\nu$. By expanding $H_{BA}$ in a Taylor series in $\nu$ about $\nu = 0$, we obtain the Hamiltonian
\begin{equation}
\begin{aligned}\label{eqn8}
H_{COM} & = \frac{1}{4m}\left(p_x^2 + p_y^2\right) + \frac{1}{4m}\left(p_\alpha^2 + p_\beta^2\right) \\
& \ \ \ \ - \frac{2\pi}{P}\left(xp_y - yp_x + \alpha p_\beta - \beta p_\alpha\right) - \frac{m^2}{2\sqrt{x^2 + y^2}} \\
& \ \ \ \ - \frac{mM/2}{\sqrt{(x + \alpha + \frac{a}{2})^2 + (y+\beta)^2}} - \frac{mM/2}{\sqrt{(x - \alpha - \frac{a}{2})^2 + (y-\beta)^2}} \\
& \ \ \ \ - \frac{mM/2}{\sqrt{(x + \alpha - \frac{a}{2})^2 + (y+\beta)^2}} - \frac{mM/2}{\sqrt{(x - \alpha + \frac{a}{2})^2 + (y-\beta)^2}} \\
& \ \ \ \ + O(\nu).
\end{aligned}
\end{equation}
We consider the problem when $\nu = 0$. We will comment on the problem when $\nu \neq 0$ in Subsection 5.4.

\subsection{Reduction}
The equations of motion in the COM coordinates $(\alpha,\beta,p_\alpha,p_\beta)$ defined by \eqref{eqn8} are
\begin{align*}
\dot{\alpha} & = \frac{p_\alpha}{2m} + \frac{2\pi \beta}{P} & 
\dot{\beta} & = \frac{p_\beta}{2m} - \frac{2\pi \alpha}{P}
\end{align*}
\begin{align}
\dot{p}_{\alpha} & = \frac{2\pi p_\beta}{P} - \frac{mM(x+\alpha + a/2)/2}{({(x + \alpha + \frac{a}{2})^2 + (y+\beta)^2})^{3/2}} + \frac{mM(x-\alpha - a/2)/2}{((x - \alpha - \frac{a}{2})^2 + (y-\beta)^2)^{3/2}}\nonumber\\
& \ \ \ \ - \frac{mM(x+\alpha - a/2)/2}{((x + \alpha - \frac{a}{2})^2 + (y+\beta)^2)^{3/2}} + \frac{mM(x-\alpha + a/2)/2}{((x - \alpha + \frac{a}{2})^2 + (y-\beta)^2)^{3/2}}\nonumber\\
\dot{p}_{\beta} & = -\frac{2\pi p_\alpha}{P} - \frac{mM(y+\beta)/2}{((x + \alpha + \frac{a}{2})^2 + (y+\beta)^2)^{3/2}} + \frac{mM(y-\beta)/2}{((x - \alpha - \frac{a}{2})^2 + (y-\beta)^2)^{3/2}}\nonumber\\
& \ \ \ \ - \frac{mM(y+\beta)/2}{((x + \alpha - \frac{a}{2})^2 + (y+\beta)^2)^{3/2}} + \frac{mM(y-\beta)/2}{((x - \alpha + \frac{a}{2})^2 + (y-\beta)^2)^{3/2}}.
\label{eqn9}
\end{align}
The zero functions $\alpha = \beta = p_\alpha = p_\beta = 0$ satisfy \eqref{eqn9}. Setting $\alpha=\beta=p_\alpha=p_\beta = 0$ corresponds to the two asteroids having a fixed barycenter at the origin, where the position and the momentum of the asteroid of mass $m_1 = m$ satisfy
\[ (x_1,y_1) = (x,y),\  (p_{x_1},p_{y_1}) = \frac{1}{2}(p_x,p_y),\]
and the position and the momentum of asteroid of mass $m_2=m$ satisfy
\[ (x_2,y_2) = -(x,y),\ (p_{x_2},p_{y_2}) = -\frac{1}{2} (p_x,p_y). \]
Thus, the two asteroids move symmetrically opposite each other through the origin, where the coordinates $(x,y,p_x,p_y)$ describe the motion of the asteroid with mass $m_1=m$. By the symmetry, the motion of the asteroid with mass $m_2=m$ is described by the coordinates $(-x,-y,-p_x,-p_y)$.

The Hamiltonian for motion of the asteroid with mass $m=m_1$ is obtained by substitution of solution $\alpha=0$, $\beta=0$, $p_\alpha=0$, and $p_\beta=0$ in the Hamiltonian $H_{COM}$. This gives the Hamiltonian
\begin{align}
H_{xy}(x,y,p_x,p_y) & = \frac{1}{4m}\left(p_x^2 + p_y^2\right) - \frac{2\pi}{P}\left(xp_y - yp_x\right) - \frac{m^2}{2\sqrt{x^2+y^2}} \nonumber\\
& \ \ \ \ - \frac{mM}{\sqrt{(x+a/2)^2 + y^2}} - \frac{mM}{\sqrt{(x-a/2)^2 + y^2}}.
\label{eqn10}
\end{align}
It is straight-forward to verify that the four equations of the Hamiltonian system associated with $H_{xy}$ agree with the equations $\dot x = \partial H_{\rm COM}/\partial p_x$, $\dot y = \partial H_{\rm COM}/\partial p_y$, $\dot p_x = - \partial H_{\rm COM}/\partial x$, and $\dot p_y = - \partial H_{\rm COM}/\partial y$ with $\alpha=\beta=p_\alpha=p_\beta=0$ substituted into them. This completes the reduction to symmetric configurations when $m_1=m_2=m$ and $M_1=M_2=M/2$ and gives the following result.

\begin{theorem}[COM Reduction]\label{thm1}
Each equilibrium or $T$-periodic solution
\[ (x(t), y(t), p_x(t), p_y(t))\] 
of the Hamiltonian system for $H_{xy}$ in \eqref{eqn10} defines an equilibrium or $T$-periodic solution of the Hamiltonian system for $H_{BA}$ in \eqref{eqn6} $($the planar Binary Asteroid Problem$)$ when $m_1 = m_2 = m$ and $M_1 = M_2 = M/2$, given in the COM coordinates \eqref{eqn7} by \[ (x(t),y(t),0,0,p_x(t),p_y(t),0,0).\]
\end{theorem}

\begin{remark} The form of the Hamiltonian $H_{xy}$ in \eqref{eqn10} is similar to that in \cite{Llibre2021}. The benefit of the COM reduction is that we may apply analysis typical of restricted problems, despite that the binary asteroid problem is not a restricted problem.
\end{remark}

\section{Relative Equilibria}
We determine which of the relative equilibria in the planar Four-Body Problem \cite{Simo} the COM reduced Binary Asteroid Problem inherits. By an analysis of the amended potential associated to the COM reduced Binary Asteroid Problem, we prove there are exactly six relative equilibria of $H_{xy}$. Four of these six relative equilibria determine two distinct symmetric collinear configurations of the four bodies. The remaining two relative equilibria determine one doubly-symmetric convex kite configuration of the four bodies, which is a rhombus. A stability analysis shows that one of the two symmetric collinear configurations is always a saddle center, independent of the relationship between $m$ and $M$, while the remaining symmetric collinear configuration is a saddle center when $m \ll M$. A stability analysis of the doubly symmetric convex kite configuration shows that it is hyperbolic when $m \ll M$.

\subsection{Amended Potential; Symmetries}
The Hamiltonian $H_{xy}$ has three parts: the kinetic term,
\[ K(p_x,p_y) = \frac{1}{4m}\left( p_x^2 + p_y^2\right),\]
the Coriolis term,
\[ \Omega(x,y,p_x,p_y) = - \frac{2\pi}{P}(xp_y - yp_x) = \frac{2\pi}{P}(yp_x - xp_y),\]
and the potential term
\[ U(x,y) = \frac{m^2}{2\sqrt{x^2+y^2}} + \frac{mM}{ \sqrt{ (x + a/2)^2 + y^2}} + \frac{mM}{\sqrt{ (x - a/2)^2 + y^2}}.\]
Thus we can write
\[ H_{xy} = K(p_x,p_y) + \Omega(x,y,p_x,p_y) - U(x,y).\]
The Hamiltonian equations associated with $H_{xy}$ then take the form
\begin{equation}\label{eqn:hameqn}
\begin{aligned} 
\dot x &  = \frac{\partial K}{\partial p_x} + \frac{\partial \Omega}{\partial p_x} = \frac{p_x}{2m} + \frac{2\pi y}{P}, \\
\dot y & = \frac{\partial K}{\partial p_y} + \frac{\partial \Omega}{\partial p_y} = \frac{p_y}{2m} - \frac{2\pi x}{P}, \\
\dot p_x & = -\frac{\partial \Omega}{\partial x} + \frac{\partial U}{\partial x} = \frac{2\pi p_y}{P} + \frac{\partial U}{\partial x}, \\
\dot p_y & = -\frac{\partial \Omega}{\partial y} + \frac{\partial U}{\partial y} = -\frac{2\pi p_x}{P} + \frac{\partial U}{\partial y}.
\end{aligned}
\end{equation}
The Newtonian equations of the motion (the system of second-order ODEs) is obtained by eliminating the momenta from the first order system of ODEs in \eqref{eqn:hameqn}. For this we have
\begin{equation}\label{eqn:secondordersystemA}
\begin{aligned}
\ddot x & = \frac{4\pi\dot y}{P} + \frac{4\pi^2 x}{P^2} + \frac{1}{2m} \frac{\partial U}{\partial x}, \\
\ddot y & = - \frac{4\pi \dot x}{P} + \frac{4\pi^2 y}{P^2} + \frac{1}{2m} \frac{\partial U}{\partial y}.
\end{aligned}
\end{equation}

Define the amended potential
\begin{equation}\label{eqn:UandV}
V(x,y) = \frac{2\pi^2}{P^2}(x^2 + y^2) + \frac{1}{2m}U(x,y).
\end{equation}
The second-order system of ODEs in \eqref{eqn:secondordersystemA} have the form
\begin{equation}\label{eqn:secondordersystem}
\begin{aligned}
\ddot x - \frac{4\pi \dot y}{P} & = \frac{4\pi^2 x}{P^2} + \frac{1}{2m} \frac{\partial U}{\partial x} = \frac{\partial V}{\partial x}, \\
\ddot y + \frac{4\pi \dot x}{P} & =  \frac{4\pi^2 y}{P^2} + \frac{1}{2m} \frac{\partial U}{\partial y} = \frac{\partial V}{\partial y}.
\end{aligned}
\end{equation}
The Jacobi integral for the second-order system \eqref{eqn:secondordersystem} is
\[ C = 2V - (\dot x^2 + \dot y^2).\]
Since by Kepler's Third Law \eqref{eqn3} there holds
\[ \frac{2\pi^2}{P^2} = \frac{1}{2}\left( \frac{M}{a^3}\right),\]
the amended potential becomes
\begin{align*}
V(x,y) & = \frac{M}{2a^3}(x^2+y^2) + \frac{m}{4\sqrt{x^2+y^2}} \\ 
& \ \ \ \ + \frac{M}{ 2\sqrt{ (x + a/2)^2 + y^2}} +  \frac{M}{2 \sqrt{ (x - a/2)^2 + y^2}}.
\end{align*}
The amended potential is expressed solely in terms of the three parameters $m$, $M$, and $a$. The amended potential $V$ has the symmetries
\begin{equation}\label{Vsymmetries}
V\circ R_1  = V,\ V\circ R_2 = V,
\end{equation}
for the involutions
\[ R_1(x,y) = (-x,y),\ R_2(x,y) = (x,-y).\]
See Figure 1 for a graph of amended potential. See Figure 2 for the zero velocity curves $V=C/2$ associated with the amended potential. For both figures the value of $m$ is exaggerated, possibly beyond $m\ll M$, to overcome the graphing resolution issue associated with the singularity of $V$ at the origin.

\begin{figure}[ht]%
\centering
\includegraphics[width=0.7\textwidth]{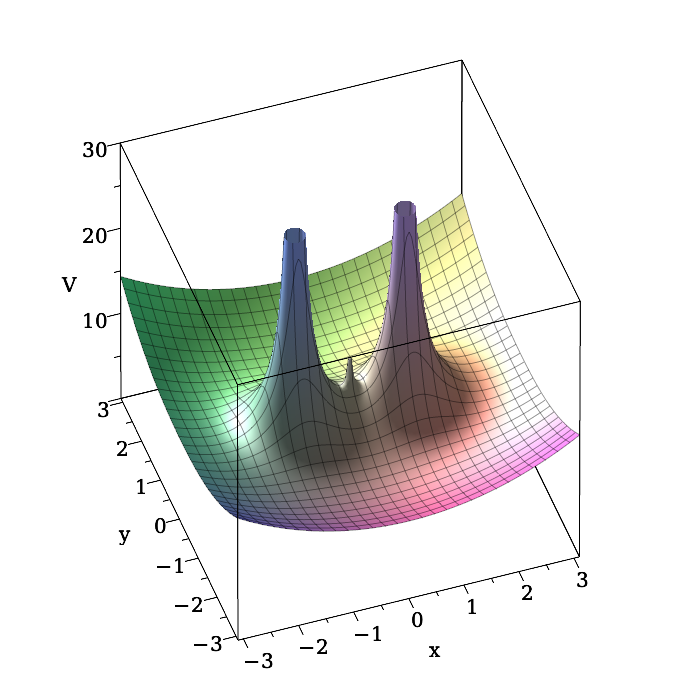}
\caption{Graph of amended potential function $V$ when $m=0.5$, $M=10$, and $a=2$.}\label{fig0A}
\end{figure}

\begin{figure}[ht]%
\centering
\includegraphics[width=0.7\textwidth]{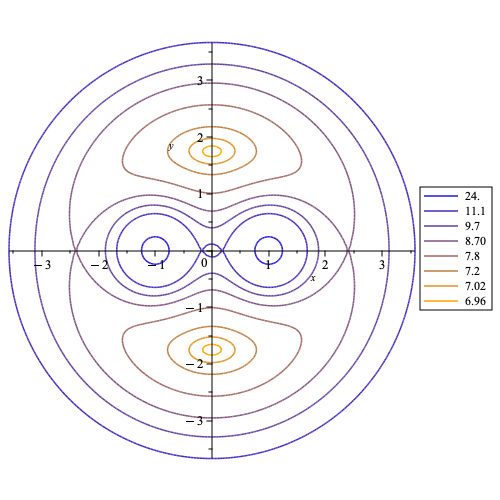}
\caption{Zero velocity curves $V = C/2$ when $m=0.5$, $M=10$, and $a=2$, where in the legend on the right is a selection of values of $C/2.$}\label{fig0B}
\end{figure}

The Hamiltonian $H_{xy}$ has the symmetries
\[ H_{xy}\circ \hat R_1 = H_{xy},\ H_{xy}\circ \hat R_2 = H_{xy},\]
for the involutions
\[ \hat R_1(x,y,p_x,p_y) = (-x,y,p_x,-p_y),\ \hat R_2(x,y,p_x,p_y) = (x,-y,-p_x,p_y),\]
where $\hat R_i$ is an extension of $R_i$, $i=1,2$. Associated to each $\hat R_i$ is a time-reversing symmetry: if $(x(t),y(t),p_x(t),p_y(t))$ is a solution of system of equations \eqref{eqn:hameqn} then so is
\begin{equation}\label{timereversingsymmetry}
\hat R_i(x(-t),y(-t),p_x(-t),p_y(-t)),\ i=1,2.
\end{equation}

\subsection{Critical Points}
Equilibria of the second-order system in \eqref{eqn:secondordersystem} are given by the critical points of the amended potential $V$. The partial derivatives of $V$ are
\begin{equation}\label{eqn:partialV}
\begin{aligned}
\frac{\partial V}{\partial x}
& = \frac{Mx}{a^3} - \frac{mx}{4(x^2+y^2)^{3/2}} - \frac{M(x+a/2)}{2((x+a/2)^2 + y^2)^{3/2}} \\
& \ \ \ \ - \frac{ M(x-a/2)}{ 2 ((x-a/2)^2 + y^2)^{3/2}}, \\
\frac{\partial V}{\partial y} 
& = y\bigg[ \frac{M}{a^3} - \frac{m}{4(x^2+y^2)^{3/2}} - \frac{M}{2((x+a/2)^2 + y^2)^{3/2}} \\
&\ \ \ \ \ \  - \frac{M}{2((x-a/2)^2 + y^2)^{3/2}}\bigg].
\end{aligned}
\end{equation}
Setting $m=0$, $M=1/2$, and $a=1$ in \eqref{eqn:partialV} gives precisely the equations that determine the five equilibria in the Circular Restricted Three-Body Problem with equal primary masses, otherwise known as the classical Copenhagen Problem \cite{meyerOffin,Szebehely}. For $m=0$, $M>0$, and $a>0$ in \eqref{eqn:partialV}, we refer to the five critical points
\[ ({\mathcal L}_1,0), (\pm{\mathcal L}_2,0),(0,\pm{\mathcal L}_4)\]
of $V$ as the equilibria of the Generalized Copenhagen Problem. Here ${\mathcal L}_1 =0$, ${\mathcal L}_2 > a/2$, and ${\mathcal L}_4 = \sqrt 3 a /2$. We will prove the existence of six equilibria
\[ (\pm L_1,0), (\pm L_2,0),(0,\pm L_4)\]
of the binary asteroid problem ($m>0$). The four equilibria $(\pm L_1,0)$ and $(\pm L_2,0)$ form collinear symmetric configurations and the two configurations $(0,\pm L_4)$ form a doubly symmetric convex kite (and rhombus) configuration. We relate the six equilibria of the binary asteroid problem to the five equilibria of the Generalized Copenhagen Problem as $m\to 0$.

The partial derivatives of $V$ imply there are two cases in which to search for critical points. Setting the partial derivative $\partial V/\partial y$ to zero gives $y=0$ and
\begin{equation}\label{eqn:secondVy}
0 = \frac{M}{a^3} - \frac{m}{4(x^2+y^2)^{3/2}} - \frac{M}{2((x+a/2)^2 + y^2)^{3/2}} - \frac{M}{2((x-a/2)^2 + y^2)^{3/2}}.
\end{equation}
Substitution of $y=0$ into $\partial V/\partial x = 0$ gives the first case,
\begin{equation}\label{eqn:firstcase}
 0 = \frac{Mx}{a^3} - \frac{mx}{4\vert x\vert^3} - \frac{M(x+a/2)}{2\vert x+a/2\vert^3}-\frac{M(x-a/2)}{2\vert x-a/2\vert^3}.
\end{equation}
To obtain the second case, we multiply \eqref{eqn:secondVy} by $-x$ and add it to $\partial V/\partial x = 0$. This gives an equation that implies $x=0$. Substitution of $x=0$ into \eqref{eqn:secondVy} gives the second case,
\begin{equation}\label{eqn:secondcase}
0 = \frac{M}{a^3} - \frac{m}{4\vert y\vert^3} - \frac{M}{((a/2)^2 + y^2)^{3/2}}.
\end{equation}

Further restrictions on the domains in the two cases follow from the symmetries of $V$ given in \eqref{Vsymmetries} . 
We need only search for critical points in first case on $x>0$, and in the second case on $y>0$.

\begin{theorem}\label{ExistenceCritPts} There exist $6$ critical points $(\pm L_1,0)$, $(\pm L_2,0)$, and $(0,\pm L_4)$ of $V$ that are smoothly dependent on $(m,M,a)\in ({\mathbb R}^+)^3$ where $L_1\in(0,a/2)$, $L_2\in(a/2,\infty)$, and $L_4\in(\sqrt 3 a/2,\infty)$. For fixed $(M,a)\in( {\mathbb R}^+)^2$ there holds $L_1\to {\mathcal L}_1$, $L_2\to {\mathcal L}_2$, and $L_4\to {\mathcal L_4} = \sqrt 3 a /2$ as $m\to 0^+$, and $L_4$ is an increasing function of $m$ with $L_4 \to \infty$ as $m\to\infty$. Furthermore, for fixed $(M,a)\in( {\mathbb R}^+)^2$, the expansion of $L_4$ in terms of $m$ near $m=0$ is
\[ L_4(m) = \frac{ \sqrt 3 a}{2} + \frac{4a}{27 M}m + O(m^2).\]
\end{theorem}

\begin{proof} We start with the first case of critical points of the form $(x,0)$, $x>0$. Setting $y=0$, the amended potential is
\[ x \mapsto V(x,0) = \frac{M}{2a^3}x^2 + \frac{m}{4 x } + \frac{M}{2(x + a/2)} + \frac{M}{2\vert x - a/2\vert},\ x>0.\]
We use coercivity and convexity of $x\mapsto V(x,0)$, as is done in the Circular Restricted Three-Body Problem, to show that there is a unique critical point of $V(x,0)$ in each of the intervals
\[ (0,a/2), (a/2,\infty).\]
The limits
\begin{align*}
& \lim_{x\to 0^+} V(x,0)  = \infty,\ 
\lim_{x\to {a/2}^-} V(x,0)  = \infty, \\
& \lim_{x\to {a/2}^+} V(x,0)  = \infty, \
\lim_{x\to\infty} V(x,0)  = \infty,
\end{align*}
imply the coercivity of $V(x,0)$ on the intervals $(0,a/2)$ and $(a/2,\infty)$.
The second derivative of $V$ with respect to $x$ evaluated at $(x,0)$ for $x\in (0,a/2) \cup (a/2,\infty)$ is
\begin{equation}\label{eqn:Vxx}
\frac{\partial^2 V}{\partial x^2}(x,0) = \frac{M}{a^3} + \frac{m}{2 x^3} + \frac{M}{(x+a/2)^3} + \frac{M}{\vert x-a/2\vert^3} > 0.
\end{equation}
Thus $x\mapsto V(x,0)$ is convex on each of the intervals $(0,a/2)$ and $(a/2,\infty)$. The unique critical point $L_1\in (0,a/2)$ of $x\mapsto V(x,0)$ satisfies \eqref{eqn:firstcase}, i.e.,
\begin{equation}\label{eqn:collinear1}
 0 = \frac{\partial V}{\partial x}(L_1,0) = \frac{M L_1}{a^3} - \frac{m}{4 L_1^2} - \frac{M}{2(L_1 + a/2)^2} + \frac{M}{2(L_1 - a/2)^2},
 \end{equation}
and the unique critical point $L_2\in (a/2,\infty)$ of $x\mapsto V(x,0)$ satisfies \eqref{eqn:firstcase}, i.e,
\begin{equation}\label{eqn:collinear2}
0 = \frac{\partial V}{\partial x}(L_2,0)  =  \frac{M L_2}{a^3} - \frac{m}{4 L_2^2} - \frac{M}{2(L_2 + a/2)^2} - \frac{M}{2(L_2 - a/2)^2}.
\end{equation}

We show that $L_1$ is a smooth function of $(m,M,a)\in({\mathbb R}^+)^3$, and for each fixed $M>0$ and $a>0$, as $m\to 0$, that $L_1$ approaches the equilibrium ${\mathcal L}_1=0$ of the Generalized Copenhagen Problem. For $x\in(0,a/2)$ set
\[ F(x,m,M,a) = \frac{\partial V}{\partial x}(x,0) = \frac{Mx}{a^3} - \frac{m}{4x^2} - \frac{M}{2(x+a/2)^2} + \frac{M}{2(x-a/2)^2}.\]
Since $F(L_1,m,M,a) = 0$ by \eqref{eqn:collinear1} and
\[ \frac{\partial F}{\partial x} (L_1,0) = \frac{\partial^2 V}{\partial x^2}(L_1,0) > 0,\]
it follows by the Implicit Function Theorem that $L_1$ is a smooth function of $(m,M,a)$ on $({\mathbb R}^+)^3$. Fixing the values of $M>0$ and $a>0$ shows that $L_1$ is a smooth function $m$ on the interval $m>0$. We cannot apply the Implicit Function Theorem to $F$ because it is not defined when $x=0$. Instead we split $F$, i.e., $\partial V/\partial x(x,0)$, into two functions,
\[ f(x,M,a) = \frac{Mx}{a^3} - \frac{M}{2(x+a/2)^2} + \frac{M}{2(x-a/2)^2}\]
and
\[ g(x,m) = \frac{m}{4x^2},\]
so that
\[ F(x,m,M,a) = f(x,M,a) - g(x,m).\]
For fixed $M>0$ and $a>0$, the function $f$ has the limits
\[ \lim_{x\to 0^+} f(x,M,a) = 0, \ \lim_{x\to a/2^-} f(x,M,a) = \infty.\]
Since
\[ \frac{\partial f}{\partial x} = \frac{M}{a^3} + \frac{M}{(x + a/2)^3 }- \frac{M}{(x-a/2)^3}\]
is positive on the interval $x\in(0,a/2)$, it follows that the function $f$ is positive and increasing on $x\in (0,a/2)$. The function $g$ is positive and decreasing on $x\in(0,a/2)$. The graphs of $f$ and $g$ intersect at precisely one point, whose $x$ value is $L_1$. As $m\to 0^+$ the graph of $f$ remains fixed while the graph of $g$ shifts the unique intersection of the graphs of $f$ and $g$ towards the origin, i.e., $L_1\to 0^+$. Setting $m=0$ in $F$ gives
\[ F(x,0,M,a) = \frac{Mx}{a^3} - \frac{M}{2(x+a/2)^2} + \frac{M}{2(x-a/2)^2},\]
a solution of which is $x=0$, the equilibrium  ${\mathcal L}_1$ of the Generalized Copenhagen Problem. Thus, the equilibrium $L_1$ converges to the equilibrium ${\mathcal L}_1$ of the Generalized Copenhagen Problem when $m\to 0^+$.

We show that $L_2$ is a smooth function of $(m,M,a)\in({\mathbb R}^+)^3$, and for fixed $M>0$ and $a>0$, that as $m\to 0^+$, $L_2$ approaches the equilibrium ${\mathcal L}_2$ of the Generalized Copenhagen Problem. For $x\in(a/2,\infty)$ now set
\[ F(x,m,M,a) =  \frac{\partial V}{\partial x}(x,0) = \frac{Mx}{a^3} - \frac{m}{4x^2} - \frac{M}{2(x+a/2)^2} - \frac{M}{2(x-a/2)^2}.\]
Since $F(L_2,m,M,a) = 0$ from \eqref{eqn:collinear2} and
\[ \frac{\partial F}{\partial x} (L_2,0) = \frac{\partial^2 V}{\partial x^2}(L_2,0) > 0\]
it follows from the Implicit Function Theorem that $L_2$ is a smooth function of $(m,M,a)$ on $({\mathbb R}^+)^3$. Fixing the values of $M>0$ and $a>0$ shows that $L_2$ is a smooth function of $m$ on the interval $m>0$. Setting $m=0$ in $F(x,m,M,a) = 0$ gives
\[ 0 = \frac{Mx}{a^3}  - \frac{M}{2(x+a/2)^2} - \frac{M}{2(x-a/2)^2}\]
for which the unique solution $x>a/2$ is the equilibrium ${\mathcal L}_2$ of the generalized Copenhagen problem. Since
\[ \frac{\partial F}{\partial x}({\mathcal L}_2,0,M,a) = \frac{M}{a^3} + \frac{M}{( {\mathcal L}_2 + a/2)^3} + \frac{M}{({\mathcal L}_2 - a/2)^3} > 0,\]
the uniqueness part of the Implicit Function Theorem implies that ${\mathcal L}_2$ is the limit of $L_2$ as $m\to 0^+$.

Now we treat the second case of the critical point of the form $(0,y)$, $y>0$. We show that $L_4$ uniquely exists and smoothly depends on $(m,M,a)\in ({\mathbb R}^+)^3$. Using equation \eqref{eqn:secondcase} set 
\[ f(y,m) = \frac{m}{4y^3} {\rm\ and\ } g(y,M,a) = \frac{M}{a^3} - \frac{M}{((a/2)^2 + y^2)^{3/2}},\]
and define
\[ F(y,m,M,a) = g(y) - f(y).\]
The function $f$ is positive with limit $\infty$ as $x\to 0^+$ and limit $0$ as $x\to\infty$, and is decreasing. The function $g$ has limits
\[ \lim_{y\to 0^+} g(y) -\frac{7M}{a^3}<0 {\rm\ and \ }
\lim_{y\to \infty} g(y) = \frac{M}{a^3} >0,\]
and is increasing because
\[ g^\prime(y) = \frac{3My}{((a/2)^2 + y^2)^{5/2}} >0.\]
For each choice of $(m,M,a)\in ({\mathbb R}^+)^3$, the graphs of $f$ and $g$ intersect in a unique point with $y = L_4>0$ which satisfies $F(L_4,m,M,a) = 0$. Since 
\[ \frac{\partial F}{\partial y}(L_4,m,M,a) = \frac{3m}{4L_4^4} + \frac{3M L_4}{((a/2)^2 + L_4^2)^{5/2}} > 0\]
it follows by the Implicit Function Theorem that $L_4$ is smooth function of $(m,M,a) \in ({\mathbb R}^+)^3$. Fixing $M>0$ and $a>0$ shows that $L_4$ is smooth function of $m$ on the interval $m>0$. 

For $M>0$ and $a>0$ fixed, the graph of $g$ is fixed, while the graph of $f$, which is independent of $M$ and $a$, shifts the point of intersection between the graphs of $f$ and $g$ to the right as $m$ increases. Thus for fixed $M>0$ and $a>0$ the value of $L_4$ increases with increasing $m$. The limit of $g$ is positive as $y\to\infty$ while the limit of $f$ is $0$ as $m\to\infty$. This implies that $L_4\to\infty$ as $m\to\infty$.

For fixed $M>0$ and $a>0$, and with $m>0$ and close to $0$, we show that the unique solution $y=L_4$ of \eqref{eqn:secondcase} satisfies $L_4>\sqrt 3 a/2$ and has limit $\sqrt 3 a/2$ as $m\to 0^+$. For $y>0$ we use Equation \eqref{eqn:secondcase} to define
\begin{equation}\label{eqn:kite}
F(y,m,M,a) = \frac{M}{a^3} - \frac{m}{4y^3}  - \frac{M}{((a/2)^2 + y^2)^{3/2}} = 0.
\end{equation}
The equation $F(y,0,M,a) = 0$ has the solution
\[ y = \frac{\sqrt 3 a}{2}\]
which corresponds to the equilateral triangle equilibrium ${\mathcal L}_4$ in the generalized Copenhagen problem. Because
\[
\frac{\partial F}{\partial y}\left( \frac{\sqrt 3 a}{2},0,M,a\right)
= \frac{3\sqrt 3 M}{2 a^4} > 0,\]
the Implicit Function Theorem implies there exists a smooth function $L_4$ defined on an open neighbourhood $O$ of $\{(0,M,a):M>0,a>0\}\subset {\mathbb R}^3$ such that $L_4(0,M,a) = \sqrt 3 a/2$ and $F(L_4(m,M,a),m,M,a) = 0$ for all $(m,M,a)\in O$. In particular, for fixed $M>0$ and $a>0$, there is $\epsilon>0$ (depending on $M$ and $a$) and a smooth function $L_4(m) = L_4(m,M,a)$ such that $L_4(0)=L_4(0,M,a) = \sqrt 3 a/2$ and $F(L_4(m),m,M,a) = 0$ for all $m\in (-\epsilon,\epsilon)$. Differentiation of $F(L_4(m),m,M,a) = 0$ with respect to $m$ followed by evaluation at $m=0$ gives
\[
\frac{dL_4}{dm}(0) =\frac{4a}{27M}.
\]
This gives the Taylor series expansion of $L_4$ in $m$ about $m=0$ as
\begin{equation}\label{eqn:L4expansion}
L_4(m) = \frac{\sqrt 3 a}{2} + \frac{4a}{27 M} m + O(m^2).
\end{equation}
Thus, for $m>0$ and close to $0$ the unique critical point $L_4(m)>\sqrt 3 a/2$ satisfies $L_4\to {\mathcal L}_4 = \sqrt 3 a /2$ as $m\to 0^+$.
\end{proof}

\begin{remark}
Even though the six equilibria $(\pm L_1,0)$, $(\pm L_2,0)$, and $(0,\pm L_4)$ exist for all $(m,M,a)\in ({\mathbb R}^+)^3$, they only are physically plausible for the model when $m \ll M$. In particular, Theorem \ref{eqn:L4expansion} states for fixed $M>0$ and $a>0$ that the pair $(0,\pm L_4)$ has the property that $L_4\to \infty$ as $m\to \infty$. This implies in the case of the binary asteroids having the same mass as the primaries, $m=M$, the configuration of the four bodies is not a square. This stands in contrast to the known square configuration in the equal mass Four-Body Problem \cite{Roy}.
\end{remark}

\begin{remark}
The expansion of $L_4(m)$ in \eqref{eqn:L4expansion} is similar to, but differs from that in Theorem 2 part (c) in \cite{Corbera2014}. The difference is that here $L_4(m)$ increases from $L_4(0)$ as $m$ increases from $0$, whereas in \cite{Corbera2014} the equivalent value of their $L_4(m)$ decreases from their $L_4(0)$ as $m$ increases from $0$ (cf Case A2 in \cite{Roy} when $m$ is small).
\end{remark}

Although Theorem \ref{ExistenceCritPts} gives six critical points, hence six equilibria of $H_{xy}$, there are only $3$ distinct configurations of the four bodies associated to the six equilibria. The primaries have the same mass $M$ and are located at the $R_1$-symmetric positions $(\pm a/2,0)$. The COM coordinates imply that placing the asteroid with mass $m_1=m$ at $(x,y)$ means placing the asteroid with mass $m_2=m$ at $(-x,-y)$. The six equilibria  $(\pm L_1,0)$, $(\pm L_2,0)$, and $(0,\pm L_4)$ are in the form of three symmetric pairs, with $(\pm L_i,0)$, $i=1,2$, being $R_1$ symmetric, and $(0,\pm L_4)$ being $R_2$ symmetric. So there are three distinct choices for the placement of the two asteroids in equilibrium.

Placing the two asteroids at $(\pm L_1,0)$ gives the first of three distinct configurations of the four bodies. The masses $M/2$, $m$, $m$, $M/2$ are collinear and located at the $R_1$-symmetric positions $(-a/2,0)$, $(-L_1,0)$, $(L_1,0)$, $(a/2,0)$, where for fixed $M>0$ and $a>0$ the value of $L_1$ depends smoothly on $m>0$ by Theorem \ref{ExistenceCritPts} (cf Section 2 in \cite{Shoaib}). We call this collinear configuration {\it non-separated} because there is not a primary between the binary asteroids. It has the property that $L_1\to {\mathcal L}_1 = 0$ as $m\to 0$ by Theorem \ref{ExistenceCritPts} (cf Case C2 in \cite{Roy}). That the two asteroids in this collinear configuration are not separated means we could not use the Implicit Function Theorem to prove that $L_1\to {\mathcal L}_1$ as $m\to 0^+$.

Placing the two asteroids at $(\pm L_2,0)$ gives the second of the distinct configuration of the four bodies. The masses $m$, $M/2$, $M/2$, $m$ are collinear and located at the $R_1$-symmetric positions $(-L_2,0)$, $(-a/2,0)$, $(a/2,0)$, $(L_2,0)$, where for fixed $M>0$ and $a>0$ the value of $L_2$ depends smoothly on $m>0$ by Theorem \ref{ExistenceCritPts} (cf Section 2 in \cite{Shoaib}). We call this collinear configuration associated to $(\pm L_2,0)$ {\it separated} because there is at least one primary (in this case both) between the binary asteroids. It has the property that $L_2\to {\mathcal L}_2 \ne 0$ as $m\to 0$ by Theorem \ref{ExistenceCritPts} (cf Case C1 in \cite{Roy}). That the two asteroids are separated means we could use the Implicit Function Theorem to prove that $L_2\to {\mathcal L}_2$ as $m\to 0^+$.

Placing the two asteroids at $(0,\pm L_4)$ gives the third of three distinct configurations. The masses $M/2$, $m$, $M/2$, $m$ are located at the $R_1$-symmetric and $R_2$-symmetric positions $(-a/2,0)$, $(0,L_4)$, $(a/2,0)$, $(0,-L_4)$, i.e., it is doubly-symmetric. This configuration is a convex kite (cf \cite{Corbera2014}), is a rhombus (cf \cite{Long2002}), and has an axes of symmetry (cf \cite{AlvarezRamirez2013}).

\subsection{Linearization Matrix}
At a equilibrium, the linearization of the Hamiltonian system in \eqref{eqn:hameqn} is
\[ A = \begin{bmatrix} 
0 & \displaystyle \frac{2\pi}{P} & \displaystyle \frac{1}{2m} & 0 \vspace{0.05in} \\
- \displaystyle \frac{2\pi}{P} & 0 & 0 & \displaystyle \frac{1}{2m} \vspace{0.05in} \\
\displaystyle \frac{\partial^2 U}{\partial x^2} & \displaystyle \frac{\partial^2U}{\partial x\partial y} & 0 & \displaystyle \frac{2\pi}{P} \vspace{0.05in} \\ 
\displaystyle \frac{\partial^2 U}{\partial x\partial y} & \displaystyle \frac{\partial^2U}{\partial y^2} & -\displaystyle \frac{2\pi}{P} & 0
\end{bmatrix}.
\]
Using the relationship between the potential $U$ and the amended potential $V$ in \eqref{eqn:UandV}, using Kepler's Third Law in \eqref{eqn3}, and setting
\[ \omega^2 = \frac{M}{a^3},\ \alpha = \frac{\partial^2 V}{\partial x^2}, \ \beta = \frac{\partial^2 V}{\partial x\partial y},\ \gamma = \frac{\partial^2 V}{\partial y^2},\]
the linearization matrix becomes
\[ 
A=\begin{bmatrix}
0 & \omega & \displaystyle \frac{1}{2m} & 0 \vspace{0.05in} \\
- \omega & 0 & 0 & \displaystyle \frac{1}{2m} \vspace{0.05in} \\
 2m (\alpha -\omega^2)   & \displaystyle 2m \beta  & 0 & \omega \vspace{0.05in} \\ 
2m \beta &  2m( \gamma -\omega^2)  & -\omega & 0
\end{bmatrix}.
\]
The characteristic polynomial of $A$ is
\[ {\rm det}(zI - A) = z^4 + (4\omega^2 - \alpha - \gamma)z^2 + \alpha\gamma - \beta^2.\]
Setting $w=z^2$, the characteristic polynomial becomes
\[ w^2 +(4\omega^2 - \alpha - \gamma) w + \alpha\gamma - \beta^2.\]
An equilibrium is an saddle-center if $\alpha\gamma-\beta^2<0$, and an equilibrium is hyperbolic if the discriminant
\[ (4\omega^2 - \alpha-\gamma)^2 - 4(\alpha\gamma - \beta^2) < 0.\]

\subsection{Spectral Stability}
By the symmetry of $V$, given in \eqref{Vsymmetries}, we need only determine the eigenvalues of the linearizations for $(L_1,0)$, $(L_2,0)$, and $(0,L_4)$. 

\begin{theorem}\label{thmSpectralStability} 
Let $(m,M,a)\in ({\mathbb R}^+)^3$. The equilibrium $(L_1,0)$ is a saddle-center, the equilibrium $(L_2,0)$ satifies $L_2\in (a/2,3a/2)$ and is a saddle-center if $m\ll M$, and the equilibrium $(0,L_4)$ is hyperbolic if $m \ll M$.
\end{theorem}

\begin{proof} To show that $(L_1,0)$ and $(L_2,0)$ are saddle-centers we show that $\alpha\gamma -\beta^2<0$. Computing the sign of $\alpha$ and $\beta$ for $L_1$ and $L_2$ is relatively easy to do. By the convexity of $x\mapsto V(x,0)$, as given in \eqref{eqn:Vxx}, it follows that
\[ \alpha = \frac{\partial^2 V}{\partial x^2}(L_i,0) > 0,\ i=1,2.\]
Since 
\begin{equation}\label{eqn:partialVxy}
\frac{\partial^2 V}{\partial x\partial y} = \frac{3mxy}{4(x^2+y^2)^{5/2}} + \frac{3M(x+a/2)y}{2((x+a/2)^2+y^2)^{5/2}} + \frac{3M(x-a/2)y}{2((x-a/2)^2+y^2)^{5/2}},
\end{equation}
it follows that
\[ \beta = \frac{\partial^2 V}{\partial x\partial y}(L_i,0) = 0, \ i=1,2.\]

To compute the sign of $\gamma$ for $L_i$, $i=1,2$, we start with the expression
\begin{equation}\label{eqn:gamma}
  \gamma = \frac{\partial^2 V}{\partial y^2}(L_i,0)  
  = \frac{M}{a^3} - \frac{m}{4L_i^3} - \frac{M}{2d_1^3} - \frac{M}{2d_2^3},
\end{equation}
where
\[ d_1 = \vert L_i + a/2\vert {\rm\ and \ } d_2 = \vert L_i - a/2\vert,\ i=1,2.\]
When $i=1$ the values $d_1$ and $d_2$ satisfy $d_1 = L_1 + a/2$, $d_2  = a/2 - L_1$, and $d_1 + d_2 = a$, and equation \eqref{eqn:collinear1} that $L_1$ satisfies becomes
\begin{equation}\label{eqn:collinear1d}
0 = \frac{M L_1}{a^3} - \frac{m}{4 L_1^2} - \frac{M}{2d_1^2} + \frac{M}{2d_2^2}.
\end{equation}
When $i=1$ the values $d_1$ and $d_2$ satisfy $d_1 = L_2 + a/2$, $d_2 = L_2 - a/2$, and $d_1 - d_2 = a$, and equation \eqref{eqn:collinear2} that $L_2$ satisfies becomes
\begin{equation}\label{eqn:collinear2d}
0  =  \frac{M L_2}{a^3} - \frac{m}{4 L_2^2} - \frac{M}{2d_1^2} - \frac{M}{2d_2^2}.
\end{equation}
Using \eqref{eqn:collinear1d} and the properties of $d_1$ and $d_2$ for $L_1$ and then \eqref{eqn:collinear2d} and the properties of $d_1$ and $d_2$ for $L_2$, and some algebra in both cases, the expression for $\gamma$ at $L_i$ in \eqref{eqn:gamma} becomes
\[ \gamma = - \frac{m}{4 L_i^2}\left[\frac{d_1 - L_i}{L_i d_1}\right] - \frac{M}{2d_1}\left[ \frac{a^3 - d_2^3}{a^2d_2^3}\right]. \]
Since $d_1-L_i = a/2 >0$ for $i=1,2$, there holds
\[ - \frac{m}{4L_i^2}\left[\frac{d_1 - L_i}{L_i d_1}\right] < 0.\]
To determine the sign of the remaining term in $\gamma$ requires separate investigations for $L_1$ and $L_2$.

To determine the value of $\gamma$ for $L_1$, we note that $d_2 = -L_1 + a/2$ and $L_1\in (0,a/2)$. So $a>d_2$ which gives $a^3>d_2^3$ and hence
\[ - \frac{M}{2d_1}\left[ \frac{a^3 - d_2^3}{a^2d_2^3}\right] < 0.\]
Thus $\gamma<0$ for $L_1$ for all $m>0$ and all $M>0$.

To determine the value of $\gamma$ for $L_2$ we note that for
\[ - \frac{M}{2d_1}\left[ \frac{a^3 - d_2^3}{a^2d_2^3}\right] < 0\]
to hold requires that
\[ a^3 - d_2^3 > 0  \ \Leftrightarrow\ a > d_2.\]
Since $d_2 = L_2 - a/2$, the requirement $a>d_2$ is equivalent to
\[ L_2 < \frac{3a}{2}.\]
Since
\[ \frac{\partial V}{\partial x}\left( \frac{3a}{2},0\right) = \frac{1}{a^2}\left[ \frac{7M}{8} - \frac{m}{9}\right],\]
and $x\mapsto V(x,0)$ is convex on $(a/2,\infty)$, the requirement $a>d_2$ is satisfied when $0<m\ll M$, i.e., the unique critical point of $V$ on $(a/2,\infty)$ occurs in interval $(a/2, 3a/2)$ because $\partial V/\partial x(3a/2,0)>0$ when $0<m\ll M$.

To show that $(0,L_4)$ is hyperbolic we show that the discriminant $(4\omega^2 - \alpha - \gamma)^2 - 4(\alpha\gamma-\beta)<0$. For $L_4$ we compute explicit simplified expressions of $\alpha$ and $\gamma$ and determine the value of $\beta$. For the latter it follows easily from \eqref{eqn:partialVxy} that $\beta = 0$ because $x=0$ for $(0,L_4)$. Computing $\alpha$ we have
\begin{equation}\label{eqn:alphaL4}
    \alpha
    =  \frac{M}{a^3} - \frac{m}{4 L_4^3} + \frac{3Ma^2}{4((a/2)^2 + L_4^2)^{5/2}} - \frac{M}{((a/2)^2 + L_4^2)^{3/2}}.
\end{equation}
The point $L_4$ satisfies \eqref{eqn:kite}, i.e.,
\begin{equation}\label{eqn:kiteL4}
\frac{M}{a^3} - \frac{m}{4 L_4^3}  - \frac{M}{((a/2)^2 + L_4^2)^{3/2}} = 0.
\end{equation}
Making the obvious substitution of \eqref{eqn:kiteL4} into \eqref{eqn:alphaL4} gives
\begin{equation}\label{eqn:alphaL4simplified}
\alpha =  \frac{3Ma^2}{4((a/2)^2 + L_4^2)^{5/2}} > 0.
\end{equation}
Evaluation of $\partial^2 V/\partial y^2$ at the point $(0,L_4)$ gives
\begin{equation}\label{gammaL4}
\gamma = \frac{M}{a^3} + \frac{2m}{4 L_4^3} + \frac{3M L_4^2}{((a/2)^2+ L_4^2)^{5/2}} - \frac{M}{((a/2)^2+L_4^2)^{3/2}}.
\end{equation}
Substitution of \eqref{eqn:kiteL4} into \eqref{gammaL4} gives
\begin{equation}\label{eqn:gammaL4simplified}
\gamma = \frac{3m}{4 L_4^3}+ \frac{3M L_4^2}{((a/2)^2+ L_4^2)^{5/2}} > 0.
\end{equation}
Substitution of the simplified expressions for $\alpha$ in \eqref{eqn:alphaL4simplified} and $\gamma$ in \eqref{eqn:gammaL4simplified} into the discriminant, and substitution of the Taylor series expansion of $L_4$ as a function of $m$ about $m=0$ from \eqref{eqn:L4expansion}, and then expanding the resulting expression as a Taylor series in $m$ about $m=0$ gives
\[ (4\omega^2 -\alpha -\gamma)^2 - 4(\alpha\gamma -\beta^2)
= \frac{M}{a^6}\left[ -\frac{23M}{4} + \frac{5\sqrt 3 m}{3}\right] + O(m^2).\]
This shows that the discriminant is negative when $0<m\ll M$.
\end{proof}

The hyperbolic dynamics near $(0,L_4)$ for $m \ll M$ are related to corresponding hyperbolic dynamics near $(0,-L_4)$ for $m \ll M$ by time reversing symmetry $\hat R_2$ in \eqref{timereversingsymmetry}. 

\begin{corollary}\label{LyapunovFamily} Let $(m,M,a)\in ({\mathbb R}^+)^3$. Emanating from each equilibrium $(\pm L_1,0)$ is a one-parameter family of periodic orbits. Emanating from each equilibrium $(\pm L_2,0)$ when $m\ll M$ is a one-parameter family of periodic orbits.
\end{corollary}

\begin{proof}
This follows by the Lyapunov Center Theorem (see \cite{meyerOffin}) for $(L_i,0)$, $i=1,2$. Applying the time-reversing symmetry $\hat R_1$ in \eqref{timereversingsymmetry} gives the result for $(-L_i,0)$, $i=1,2$.
\end{proof}

\section{Periodic Orbits Near the Origin}
The two asteroids in COM coordinates experience a binary collision at the origin. We show near the origin that there are two families of periodic orbits in the COM reduced Binary Asteroid Problem through symplectic scaling and the Poincar\'e Continuation Method as implemented in Meyer and Offin \cite{meyerOffin}.

\subsection{Scaling}
For small $\delta > 0$ the transformation $(x,y,p_x,p_y) \mapsto (\xi,\eta,p_\xi,p_\eta)$ given by the inverse
\begin{align}
x = \delta^2\xi, & & y = \delta^2\eta, & & p_x = \delta^{-1}p_\xi, & & p_y = \delta^{-1}p_\eta
\label{eqn:scalingdoublecollision}
\end{align}
is symplectic with multiplier $\delta^{-1}$. The Hamiltonian $H_{xy}$ in \eqref{eqn10} becomes
\begin{equation}
\begin{aligned}
H_{\xi\eta}
& = \delta^{-1} H_{xy}  \\
& = \frac{1}{4\delta^3 m}\big( p_\eta^2 + p_\xi^2\big) - \frac{2\pi}{P}\big(\xi p_\eta - \eta p_\xi\big) - \frac{m^2}{ 2\delta^3 \sqrt { \xi^2 + \eta^2}}  \\
&  \ \ \ \  -\frac{ m M}{\delta \sqrt { (\delta^2\xi + a/2)^2 + (\delta^2\eta)^2}} - \frac{mM}{ \delta \sqrt{ (\delta^2 \xi - a/2)^2 + (\delta^2 \eta)^2}}.
\label{eqn:HamDC1}
\end{aligned}
\end{equation}
Scaling time by $t\to \delta^{-3}t$ is the same as multiplying $H_{\xi\eta}$ by $\delta^3$. Multiplying $H_{\xi\eta}$ in \eqref{eqn:HamDC1} by $\delta^3$ gives the Hamiltonian
\begin{equation}
\begin{aligned} 
H_{\xi\eta}
& =  \frac{1}{4m}\big( p_\eta^2 + p_\xi^2\big) - \frac{2\pi\delta^3}{P}\big(\xi p_\eta - \eta p_\xi\big) - \frac{m^2}{ 2 \sqrt { \xi^2 + \eta^2}} \\
&  \ \ \ \  - \frac{\delta^2 m M}{ \sqrt { (\delta^2\xi + a/2)^2 + (\delta^2\eta)^2}} - \frac{\delta^2 mM}{\sqrt{ (\delta^2 \xi - a/2)^2 + (\delta^2 \eta)^2}}.
\label{eqn:HamDC2}
\end{aligned}
\end{equation}
The last two terms in \eqref{eqn:HamDC2} are real analytic in $\delta$ at $\delta=0$. Their Taylor series expansions in $\delta$ about $\delta = 0$ are
\begin{align*}
-\frac{\delta^2 m M}{ \sqrt { (\delta^2\xi + a/2)^2 + (\delta^2\eta)^2}}
& = - \frac{2mM}{a} \delta^3 + \frac{ 4mM\xi}{ a^2} \delta^5 + O(\delta^7), \\
- \frac{\delta^2 mM}{\sqrt{ (\delta^2 \xi - a/2)^2 + (\delta^2 \eta)^2}}
& = - \frac{2mM}{a} \delta^3 - \frac{ 4mM\xi }{ a^2} \delta^5 + O(\delta^7).
\end{align*}
In these two expansions the terms of order $\delta^5$ cancel and the terms of order $\delta^3$ are constants (which we ignore), so the expansion of the Hamiltonian is
\[ H_{\xi\eta} = H_0 + \delta^3 H_3 + O(\delta^7)\]
where
\begin{align*}
H_0 & = \frac{1}{4m}\big( p_\xi^2 + p_\eta^2\big) - \frac{m^2}{2\sqrt{\xi^2 + \eta^2}}, \\
H_3 & = -\frac{2\pi}{P}\big( \xi p_\eta - \eta p_\xi\big).
\end{align*}

\subsection{Periodic Solutions to $O(\delta^3)$}
Seeking circular periodic solutions, we transform $H = H_0 + \delta^3 H_3 + O(\delta^7)$ into symplectic polar coordinates
\[ (\xi,\eta,p_\xi,p_\eta) \to (r,\theta,p_r,p_\theta)\]
defined by (the inverse transformation, see \cite{meyerOffin})
\begin{align*}
\xi & = r\cos\theta,\ \eta = r\sin\theta, \\
p_\xi & = p_r \cos\theta - \frac{p_\theta}{r}\sin\theta,\ p_\eta  = p_r\sin\theta + \frac{p_\theta}{r}\cos\theta.
\end{align*}
It follows that
\[ H_0  = \frac{1}{4m}\left[ p_r^2 + \frac{p_\theta^2}{r^2}\right] - \frac{m^2}{2r},\ 
H_3 = - \frac{2\pi p_\theta}{P}.\]
The corresponding system of ODEs for $H_0 + \delta^3 H_3$ is
\begin{align} \label{eqn:HamOrigin}
\dot r  = \frac{p_r}{2m}, & & 
\dot \theta  = \frac{p_\theta}{2mr^2} - \delta^3 \frac{2\pi}{P}, & &
\dot p_r = \frac{p_\theta^2}{2mr^3} - \frac{m^2}{2r^2}, & &
\dot p_\theta  = 0.
\end{align}
For $c = \sqrt{m^3}$, equations \eqref{eqn:HamOrigin} have two circular periodic solutions
\begin{equation}\label{Eqn:SolOrigin}
r(t) = 1,\ \theta(t) = \theta(0) + \omega t,\ p_r(t) = 0,\ p_\theta = \pm c,
\end{equation}
whose frequencies are
\[ \omega = \pm \frac{\sqrt{m}}{2} - \delta^3 \frac{2\pi}{P} \]
and whose corresponding periods are
\[ T = \frac{2\pi}{\vert \omega\vert} = \frac{4\pi P}{\sqrt{m}P \mp 4\pi \delta^3}.\]

\subsection{Multipliers and Continuation}
We linearize the ODEs for $\dot r$ and $\dot p_r$ along the two circular periodic solutions \eqref{Eqn:SolOrigin} to obtain their multipliers.
The linearization is
\[ \dot r  = \frac{p_r}{2m}, \ \dot p_r  = -\frac{m^2}2 r.\]
The nontrivial multipliers of the two circular periodic solutions are
\[ \exp\left(\pm  i T \frac{\sqrt{m}}2 \right) = 1 \pm 2\pi i \left( \frac{4\pi}{\sqrt{m}P}\right) \delta^3 + O(\delta^6). \]
Thus, for small positive $\delta$, the two circular periodic solutions are elementary. By the argument in Meyer and Offin \cite{meyerOffin} and Theorem \ref{thm1} we have the following result.

\begin{theorem}\label{Thm:CentralHills}
There exist two one-parameter families of nearly circular periodic solutions to the Binary Asteroid Problem \eqref{eqn6} in the case of equal asteroid $(m_1 = m_2 = m)$ and equal primary $(M_1 = M_2 = M/2)$ masses, where the two asteroids are near the origin.
\end{theorem}

Figure \ref{fig:CentralHills} shows a solution to \eqref{Eqn:SolOrigin} undergoing a numerical process of continuation into \eqref{eqn10}. Notice how the orbits maintain their near-circular shape and periodic nature. This is consistent with Theorem \ref{Thm:CentralHills}.

\begin{figure}[ht]%
\centering
\includegraphics[width=1.0\textwidth]{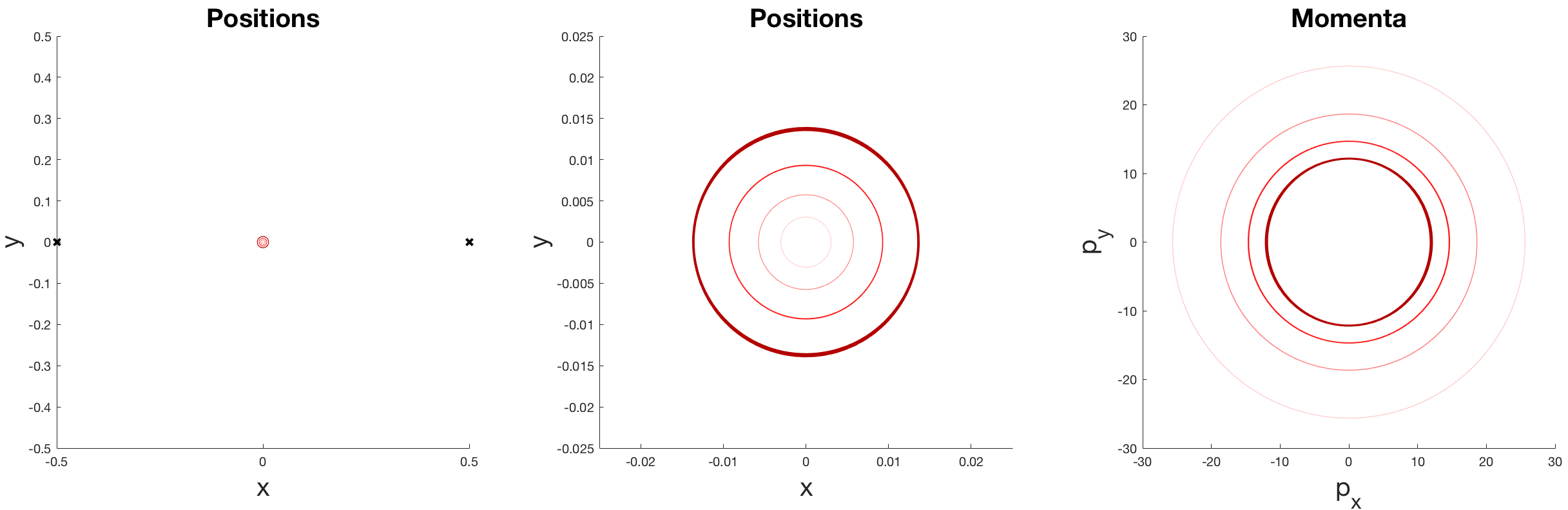}
\caption[Central-Hill continuation]{Depiction of a numerical $\delta$ continuation process for a periodic orbit near the origin, according to \eqref{Eqn:SolOrigin}. The hues darken as $\delta$ is increased, with the darkest orbit corresponding to $\delta = 1$. (Left) A view of positions showing the primaries as black x's. (Middle) The left panel zoomed in to see the detail.}\label{fig:CentralHills}
\end{figure}


\section{Hill-Type Orbits}
We show that there exist two one-parameter families of near-circular Hill-type orbits (cf. \cite{Stoica,meyerOffin}) in the COM reduced Binary Asteroid Problem where one asteroid orbits one primary and the other asteroid orbits the other primary.
We use symplectic scaling to make the Hamiltonian $H_{xy}$ resemble that of a rotating 2-body problem, which has simple circular periodic solutions.

\subsection{Scaling}

Using the COM reduced Hamiltonian \eqref{eqn10}, we perform a symplectic shift which places the left primary, located at $(x,y) = (-a/2, 0)$, at the origin. We must also shift the vertical momentum in order to mitigate the introduction of unfavorable terms. The shift is
\begin{align}
x \mapsto x + \frac{a}{2}, & & y \mapsto y, & & p_x \mapsto p_x, & & p_y \mapsto p_y + \frac{2\pi am}{P}.
\label{eqn11}
\end{align}
In the shifted coordinates, the Hamiltonian \eqref{eqn10} becomes
\begin{align}
H_{xy} & = \frac{1}{4m}\left(p_x^2 + p_y^2\right) - \frac{2\pi}{P}\left(xp_y - yp_x\right) + \frac{4\pi^2 amx}{P^2} - \frac{\pi^2a^2m}{P^2}  \nonumber\\
& \ \ \ \ - \frac{m^2}{2\sqrt{(x-a/2)^2 + y^2}} - \frac{mM}{\sqrt{x^2+y^2}} - \frac{mM}{\sqrt{(x-a)^2 + y^2}}.
\label{eqn12}
\end{align}
We use Kepler's Third Law \eqref{eqn3} to simplify $4\pi^2 amx/P^2 = mMx/a^2$. Additionally, since \eqref{eqn12} is time-independent, $H_{xy}$ remains an integral, and so the constant $-\pi^2a^2m/P^2$ term is dropped. 

For small $\delta > 0$ the transformation $(x,y,p_x,p_y) \mapsto (\xi,\eta,p_\xi,p_\eta)$ given by
\begin{align}
x = \delta^2\xi, & & y = \delta^2\eta, & & p_x = \delta^{-1}p_\xi, & & p_y = \delta^{-1}p_\eta
\label{eqn13}
\end{align}
is symplectic with multiplier $\delta^{-1}$. The Hamiltonian \eqref{eqn12} becomes
\begin{align}
    H_{\xi\eta} = \delta^{-1}H_{xy} & = \frac{\delta^{-3}}{4m}\left(p_\xi^2 + p_\eta^2\right) - \frac{2\pi}{P}\left(\xi p_\eta - \eta p_\xi\right) + \frac{\delta mM\xi}{a^2} - \frac{\delta^{-3}mM}{\sqrt{\xi^2+\eta^2}} \nonumber\\
& \ \ \ \ - \frac{\delta^{-1}m^2}{2\sqrt{(\delta^2\xi-a/2)^2 + (\delta^2\eta)^2}} - \frac{\delta^{-1}mM}{\sqrt{(\delta^2\xi-a)^2 + (\delta^2\eta)^2}}.
\label{eqn14}
\end{align}
Scaling time via $t \mapsto \delta^{-3}t$ is the same as multiplying the Hamiltonian $H_{\xi\eta}$ in \eqref{eqn14} by $\delta^3$. After multiplying $H_{\xi\eta}$ by $\delta^3$ and grouping terms of like powers of $\delta$, we find
\begin{align}
    H_{\xi\eta} & = \left[\frac{1}{4m}\left(p_\xi^2 + p_\eta^2\right) - \frac{mM}{\sqrt{\xi^2+\eta^2}}\right] - \delta^3\left[\frac{2\pi}{P}\left(\xi p_\eta - \eta p_\xi\right)\right] + \delta^4\left[\frac{mM\xi}{a^2}\right] \nonumber\\
& \ \ \ \ - \delta^2\left[\frac{m^2}{2\sqrt{(\delta^2\xi-a/2)^2 + (\delta^2\eta)^2}} + \frac{mM}{\sqrt{(\delta^2\xi-a)^2 + (\delta^2\eta)^2}}\right].
\label{eqn15}
\end{align}
The Hamiltonian in \eqref{eqn15} is analytic in $\delta$ about $0$. Hence we expand $H_{\xi\eta}$ in a Taylor series in $\delta$ about $0$, yielding
\begin{equation}
    H_{\xi\eta} = H_0 - \frac{\delta^2m(m+M)}{a} + \delta^3H_3 + \delta^4H_4 + O(\delta^6),
    \label{eqn16}
\end{equation}
where
\begin{align}
H_0 = \frac{1}{4m}\left(p_\xi^2 + p_\eta^2\right) - \frac{mM}{\sqrt{\xi^2+\eta^2}}, & & H_3 = -\frac{2\pi}{P}\left(\xi p_\eta - \eta p_\xi\right), & & H_4 = -\frac{2m^2\xi}{a^2}.
\label{eqn17}
\end{align}
Once again, we ignore the constant $-\delta^2m(m+M)/a$ by time-independence.

\subsection{Periodic Solutions to $O(\delta^3)$}
The Hamiltonian system defined by $H_0 + \delta^3H_3$ in \eqref{eqn16} has circular solutions. To obtain them, we first make the change $(\xi,\eta,p_\xi,p_\eta) \mapsto (r,\theta,p_r,p_\theta)$ to symplectic polar coordinates:
\begin{align*}
\xi & = r\cos \theta,\ \eta = r\sin\theta, \\
p_\xi & = p_r\cos\theta - \frac{p_\theta}{r}\sin\theta,\ p_\eta = p_r\sin\theta + \frac{p_\theta}{r}\cos\theta.
\end{align*}
The Hamiltonian becomes
\begin{equation}
    H_0 + \delta^3H_3 = \frac{1}{4m}\left(p_r^2 + \frac{p_\theta^2}{r^2}\right) - \frac{mM}{r} - \delta^3\frac{2\pi p_\theta}{P}.
    \label{eqn18}
\end{equation}
The equations of motion defined by \eqref{eqn18} are
\begin{align}
    \dot{r} = \frac{p_r}{2m}, & & \dot{\theta} = \frac{p_\theta}{2mr^2} - \frac{2\pi\delta^3}{P}, & & \dot{p}_r = \frac{p_\theta^2}{2mr^3} - \frac{mM}{r^2}, & & \dot{p}_\theta = 0.
    \label{eqn19}
\end{align}
For $c = \sqrt{2m^2M}$, equations \eqref{eqn19} have two circular periodic solutions
\begin{equation}\label{eqn20}
r = 1,\ \theta(t) = \theta(0) + \omega t,\ p_r = 0,\ p_\theta = \pm c,
\end{equation}
whose frequencies are
\[ \omega = \frac{\pm cP - 4\pi m \delta^3}{2mP} = \frac{\pm m\sqrt{2M}P - 4\pi m \delta^3}{2m P} \]
and corresponding periods are
\begin{equation}\label{eqn21}
 T = \frac{2\pi}{\vert \omega\vert} = \frac{ 4\pi m P}{ cP \mp 4\pi m\delta^3} = \frac{4\pi m P}{ m\sqrt{2M}P \mp 4\pi m \delta^3}.
\end{equation}

\subsection{Multipliers and Continuation}
We linearize the ODEs in $r$ and $p_r$ in \eqref{eqn19} about the two circular periodic solutions \eqref{eqn20} to obtain their multipliers. The linearization gives
\[ \dot r = \frac{p_r}{2m},\ \dot p_r = -mM r,\]
whose solutions have the form $\exp(\pm i t\sqrt{M/2})$. The nontrivial multipliers of the two circular periodic solutions are
\begin{align*} 
\exp(\pm i T\sqrt{M/2})
& =  \exp\left( \pm i \sqrt{ \frac{M}{2}} \frac{4\pi m P}{m\sqrt{2M}P \mp 4\pi m \delta^3}\right) \\
& = 1 \pm 2\pi i \left( \frac{2\pi \sqrt{2}} {P\sqrt{M}}\right)\delta^3 + O(\delta^6).
\end{align*}
Thus, for small positive $\delta$, the two circular periodic solutions are elementary. By the Poincar\'e continuation argument in Meyer and Offin \cite{meyerOffin}, and Theorem \ref{thm1} we have the following result.

\begin{theorem}[Hill Orbits]\label{thm2}
There exist two one-parameter families of near-circular Hill-type periodic solutions to the Binary Asteroid Problem \eqref{eqn6} in the case of equal asteroid $(m_1 = m_2 = m)$ and equal primary  $(M_1 = M_2 = M/2)$ masses, where each asteroid orbits near one primary.
\end{theorem}

Figure \ref{fig1} shows a solution to \eqref{eqn20} undergoing a numerical process of continuation into \eqref{eqn10}, similar to the previous section. Notice the geometric similarities with the orbits depicted in Figure \ref{fig:CentralHills}.

\begin{figure}[ht]%
\centering
\includegraphics[width=1.0\textwidth]{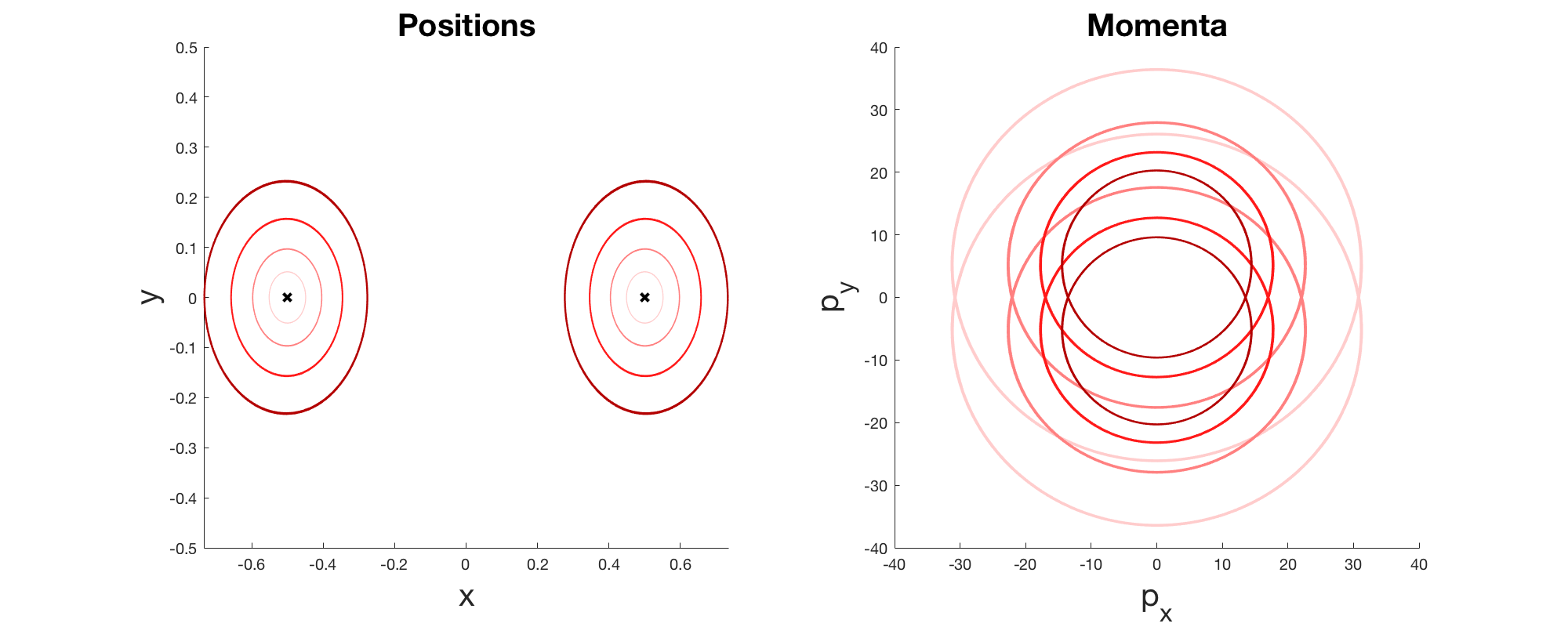}
\caption{Depiction of a numerical $\delta$ continuation process of a Hill-type orbit from the circular orbit in \eqref{eqn20}.
}\label{fig1}
\end{figure}

\subsection{Numerical Continuation}
We remark that we have some numerical evidence that the orbits of Theorem \ref{thm2} persist as quasi-periodic orbits, as we perturb the primary mass deviation parameter $\nu$ from $0$ (see Figure \ref{fig2}). This suggests that the COM approach can be adapted to perform the reduction in the arbitrary primary mass case. However, we have not yet determined what that adaptation is.

\begin{figure}[ht]%
\centering
\includegraphics[width=1.0\textwidth]{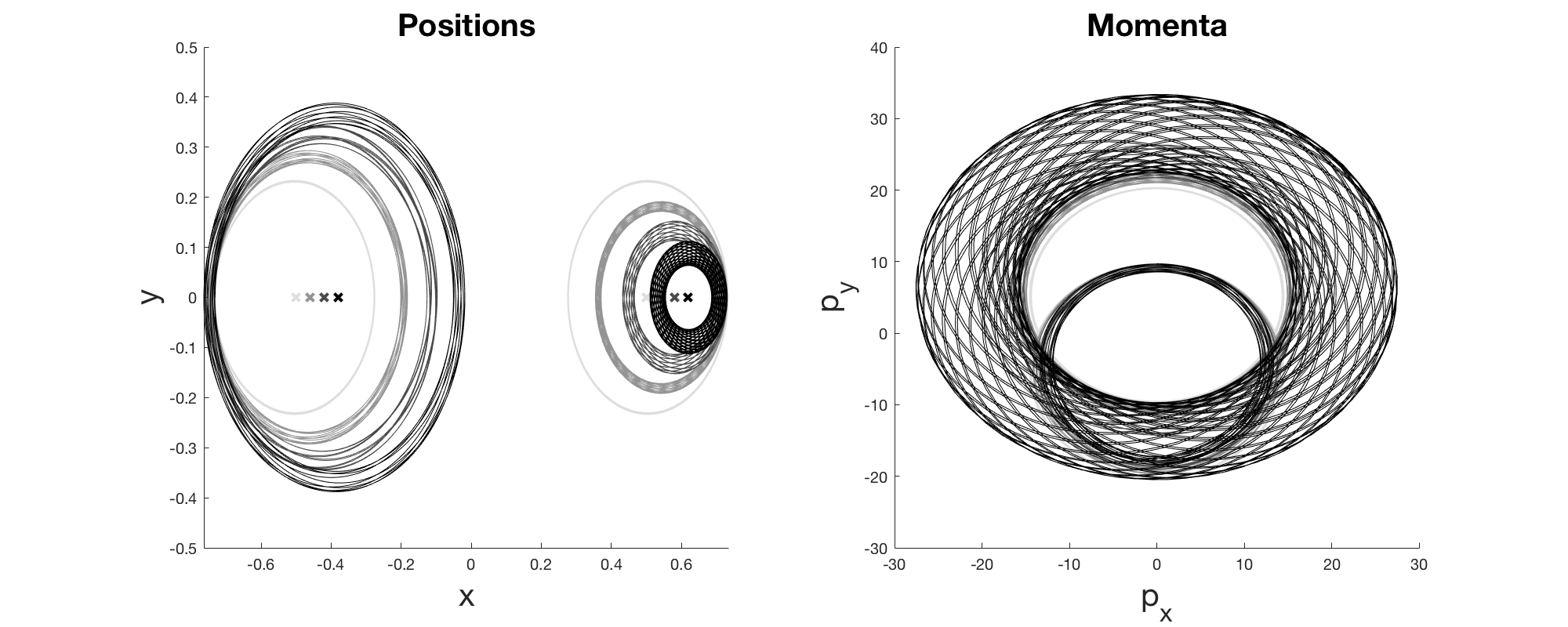}
\caption{Depiction of a numerical $\nu$ continuation process for the darkest Hill-type orbit in Figure \ref{fig1}. The x's in the left graph mark the primaries and are shaded to match their associated orbits.}\label{fig2}
\end{figure}


\section{Comet-Type Orbits}
We show that there are two one-parameter families of near-circular comet-orbits in the Binary Asteroid Problem (cf. \cite{meyerOffin,Stoica}) when the two ratios $m_1/M_1$ and $m_2/M_2$ are small enough. We could have applied a symplectic scaling to the Hamiltonian $H_{xy}$ of the COM reduced Binary Asteroid Problem, but we found a more general approach for the Hamiltonian $H_{BA}$ of the Binary Asteroid Problem that does not requires the asteroids nor the primaries to have equal mass.

\subsection{Scaling}
Following the comet approach in \cite{meyerOffin}, let $\delta > 0$ be a small scaling parameter and define the transformation $(x_i,y_i,p_{x_i},p_{y_i}) \mapsto (\xi_i,\eta_i,p_{\xi_i},p_{\eta_i})$, $i=1,2$, by
\begin{align}
x_i = \delta^{-2}\xi_i, & & y_i = \delta^{-2}\eta_i, & & p_{x_i} = \delta p_{\xi_i}, & & p_{y_i} = \delta p_{\eta_i}.
\label{eqn22}
\end{align}
This is symplectic with multiplier $\delta$ and the Binary Asteroid Hamiltonian $H_{BA}$ in \eqref{eqn6} becomes
\begin{align}
H_{\xi\eta} 
& = \delta H_{BA} \\
& = \frac{\delta^3}{2m_1}\left(p_{\xi_1}^2 + p_{\eta_1}^2\right) + \frac{\delta^3}{2m_2}\left(p_{\xi_2}^2 + p_{\eta_2}^2\right) - \frac{2\pi}{P}\left(\xi_1p_{\eta_1} - \eta_1p_{\xi_1}\right) \nonumber\\
& \ \ \ \ - \frac{2\pi}{P}\left(\xi_2p_{\eta_2} - \eta_2p_{\xi_2}\right) - \frac{\delta^3m_1m_2}{\sqrt{(\xi_1 - \xi_2)^2 + (\eta_1 - \eta_2)^2}} \nonumber\\
& \ \ \ \ -\frac{\delta^3m_1M_1}{\sqrt{(\xi_1 + a\delta^2M_2/M)^2 + \eta_1^2}} -\frac{\delta^3m_2M_1}{\sqrt{(\xi_2 + a\delta^2M_2/M)^2 + \eta_2^2}} \nonumber\\
& \ \ \ \ -\frac{\delta^3m_1M_2}{\sqrt{(\xi_1 - a\delta^2M_1/M)^2 + \eta_1^2}} -\frac{\delta^3m_2M_2}{\sqrt{(\xi_2 - a\delta^2M_1/M)^2 + \eta_2^2}}.
\label{eqn25}
\end{align}
The Hamiltonian $H_{\xi\eta}$ in \eqref{eqn25} is analytic in $\delta$ at zero, hence we expand the last four terms in $H_{\xi\eta}$ in a Taylor series in $\delta$ about $\delta=0$ and then group terms with like powers of $\delta$. The terms from the expansions with $\delta^5$ cancel and the terms with $\delta^3$ combine to give 
\begin{equation}
H_{\xi\eta} = H_0 + \delta^3H_3 + O(\delta^7),
\label{eqn26}
\end{equation}
where
\begin{align}
H_0 & = -\frac{2\pi}{P}\left(\xi_1p_{\eta_1} - \eta_1p_{\xi_1} + \xi_2p_{\eta_2} - \eta_2p_{\xi_2}\right),\nonumber\\
H_3 & = \frac{1}{2m_1}\left(p_{\xi_1}^2 + p_{\eta_1}^2\right) + \frac{1}{2m_2}\left(p_{\xi_2}^2 + p_{\eta_2}^2\right) - \frac{m_1M}{\sqrt{\xi_1^2 + \eta_1^2}} - \frac{m_2M}{\sqrt{\xi_2^2 + \eta_2^2}} \nonumber\\
& \ \ \ \ - \frac{m_1m_2}{\sqrt{(\xi_1 - \xi_2)^2 + (\eta_1 - \eta_2)^2}}.
\label{eqn27}
\end{align}

\subsection{Phase Condition in $O(\delta^3)$}
To find circular periodic solutions of the system with Hamiltonian $H_0 + \delta^3H_3$ we change to double symplectic polar coordinates $(\xi_i, \eta_i, p_{\xi_i}, p_{\eta_i}) \mapsto (r_i, \theta_i, p_{r_i}, p_{\theta_i})$. This transformation is symplectic since the new coordinates of a single asteroid (i.e., fixed $i$) depend only on that asteroid's old coordinates, via a symplectic map (see \cite{meyerOffin}).
The Hamiltonian becomes
\begin{align}
H_0 + \delta^3H_3 & = -\frac{2\pi}{P}\left(p_{\theta_1} + p_{\theta_2}\right) + \frac{\delta^3}{2m_1}\left(p_{r_1}^2 + \frac{p_{\theta_1}^2}{r_1^2}\right) + \frac{\delta^3}{2m_2}\left(p_{r_2}^2 + \frac{p_{\theta_2}^2}{r_2^2}\right) \nonumber\\
& \ \ \ \ - \frac{\delta^3m_1M}{r_1} -\frac{\delta^3m_2M}{r_2} - \frac{\delta^3m_1m_2}{\sqrt{r_1^2 + r_2^2 - 2r_1r_2\cos(\theta_1 - \theta_2)}}.
\end{align}
The associated equations of motion are
\begin{equation}
\begin{aligned}
\dot{r}_1 & = \frac{\delta^3p_{r_1}}{m_1},\ \dot{\theta}_1 = -\frac{2\pi}{P} + \frac{\delta^3p_{\theta_1}}{m_1 r_1^2}, \\
\dot{r}_2  & = \frac{\delta^3p_{r_2}}{m_2},\ \dot{\theta_2} = -\frac{2\pi}{P} + \frac{\delta^3p_{\theta_2}}{m_2 r_2^2}, \\
\dot{p}_{r_1} & = \delta^3\left[\frac{p_{\theta_1}^2}{m_1 r_1^3} - \frac{m_1 M}{r_1^2} - \frac{m_1m_2(r_1 - r_2\cos(\theta_1 - \theta_2))}{\sqrt{r_1^2 + r_2^2 - 2r_1r_2\cos(\theta_1 - \theta_2)}^3}\right], \\
\dot{p}_{\theta_1} & = -\frac{\delta^3m_1m_2r_1r_2\sin(\theta_1 - \theta_2)}{\sqrt{r_1^2 + r_2^2 - 2r_1r_2\cos(\theta_1 - \theta_2)}^3}, \\
\dot{p}_{r_2} & = \delta^3\left[\frac{p_{\theta_2}^2}{m_2 r_2^3} - \frac{m_2 M}{r_2^2} - \frac{m_1m_2(r_2 - r_1\cos(\theta_1 - \theta_2))}{\sqrt{r_1^2 + r_2^2 - 2r_1r_2\cos(\theta_1 - \theta_2)}^3}\right], \\
\dot{p}_{\theta_2} & = \frac{\delta^3m_1m_2r_1r_2\sin(\theta_1 - \theta_2)}{\sqrt{r_1^2 + r_2^2 - 2r_1r_2\cos(\theta_1 - \theta_2)}^3}.
\label{eqn29}
\end{aligned}
\end{equation}
Any solution of system \eqref{eqn29} which has $\sin(\theta_1 - \theta_2) \equiv 0$ will necessarily have $p_{\theta_i} = const$, $i=1,2$. Requiring this restriction implies $\theta_1 - \theta_2 = n\pi$ for some $n \in \Int$ and for all time. When $n$ is even, the asteroids have equal phase relative to the barycenter. We determined numerically that this behavior leads to a collision between the asteroids. However, when $n$ is odd, the asteroids have opposite phase, and there exist non-collision solutions.

In the odd case it is sufficient to study $n = 1$; we impose $\theta_1 - \theta_2 = \pi$. This implies through the $p_{\theta_i}$ ODEs that $p_{\theta_i} = c_i$ for some constants $c_i \in \Real$. The phase condition also implies that $\dot{\theta}_1 = \dot{\theta_2}$ for all time. That is,
\begin{align*}
-\frac{2\pi}{P} + \frac{\delta^3p_{\theta_1}}{m_1r_1^2} = -\frac{2\pi}{P} + \frac{\delta^3p_{\theta_2}}{m_2r_2^2} 
\implies \left(\frac{r_1}{r_2}\right)^2 = \frac{c_1/m_1}{c_2/m_2} = \left(\frac{c_1}{c_2}\right)\left(\frac{m_2}{m_1}\right) > 0.
\end{align*}
Immediately this implies $r_1/r_2 = const.$, and $c_1$ and $c_2$ must have the same sign. Define
\begin{equation}
    b = \sqrt{\frac{c_1/m_1}{c_2/m_2}} > 0,
    \label{eqn30}
\end{equation}
so that we have $r_1(t) = br_2(t)$ for all time. We can therefore solve the ODEs \eqref{eqn29} in terms of the second asteroid's coordinates $(r_2,\theta_2,p_{r_2},p_{\theta_2})$.

We investigate circular periodic solutions in the regime wherein $r_i = const$. In this case, the $r_i$ ODEs force $p_{r_i} \equiv 0$. Since the phase condition gives us that $\cos(\theta_1 - \theta_2) = \cos(\pi) = -1$, this in turn gives (after some algebra)
\begin{align*}
0 = \dot{p}_{r_1} = \frac{\delta^3m_1}{r_2^2}\left[\frac{b(c_2/m_2)^2}{r_2} - \frac{M}{b^2} - \frac{m_2}{(1+b)^2}\right].
\end{align*}
Similarly, 
\begin{align*}
0 = \dot{p}_{r_2} = \frac{\delta^3m_2}{r_2^2}\left[\frac{(c_2/m_2)^2}{r_2} - M - \frac{m_1}{(1+b)^2}\right].
\end{align*}
We can eliminate $r_2$ and the angular momentum $c_2$ from these equations if we divide each equation by $\delta^3m_i/r_2^2$ and multiply the $\dot{p}_{r_2}$ equation by $b$. Doing so gives
\begin{align*}
0 & = \left[\frac{b(c_2/m_2)^2}{r_2} - \frac{M}{b^2} - \frac{m_2}{(1+b)^2}\right] -  \left[\frac{b(c_2/m_2)^2}{r_2} - Mb - \frac{m_1b}{(1+b)^2}\right]\\
& = Mb + \frac{m_1b}{(1+b)^2} - \frac{M}{b^2} - \frac{m_2}{(1+b)^2} \\
& = \frac{Mb^3(1+b)^2 + m_1b^3 - M(1+b)^2 - m_2b^2}{b^2(1+b)^2}.
\end{align*}
By setting the numerator to zero, dividing by $M$, and collecting like powers of $b$, we find that
\begin{equation}
q(b) = 0, \text{ where } q(z) = z^5 + 2z^4 + \left(1 + \frac{m_1}{M}\right)z^3 - \left(1 + \frac{m_2}{M}\right)z^2 - 2z - 1.
\label{eqn31}
\end{equation}

\subsection{Analysis of $b$}
For most values of $m_1$, $m_2$, and $M$, we cannot factor the real polynomial $q(z)$ since it is a quintic. However, since the masses are all positive, we can analyze its roots. Descartes' Rule of Signs says that in this case, there is exactly one positive real root, for there is one sign change in the sequence of coefficients of $q(z)$. Since $b > 0$, it must be the one positive root.

In addition, we see that $q(0) = -1$, $q(z) \to \infty$ as $z \to \infty$, and $q(1) = (m_1 - m_2)/M$. By continuity of $q$, we may therefore invoke the Intermediate Value Theorem to place $b$ within $(0, \infty)$ according to the relative sizes of $m_1$ and $m_2$:
\begin{equation}
b \in  \begin{cases} 
          (0,1) & \text{if } m_1>m_2 \\
          \{1\} & \text{if } m_1 = m_2 \\
          (1, \infty) & \text{if } m_1 < m_2
       \end{cases}.
\label{eqn32}
\end{equation}
Note that \eqref{eqn32} means the case $m_1 = m_2$ forces $r_1 = r_2$, in accordance with the COM approach.

In practice, we must obtain $b$ numerically once values of $m_1$, $m_2$, and $M$ are fixed. We found that for all realistic masses (e.g. when $m_1,m_2 \ll M$), $b$ remains close to 1 (see Figure \ref{fig3a}). To get $b > 2$ or $b < 0.5$, we find that the mass ratios $\mu_i = m_i/M$ need to be on the order of 20, which is extremely unrealistic.

Figure \ref{fig3} shows the extreme case of a comet orbit with $b = 2$, to better demonstrate the effect on the asteroid trajectories. For more realistic asteroid masses, the effect is too small to observe on these graph scales.

\begin{figure}[ht]%
\centering
\includegraphics[width=0.8\textwidth]{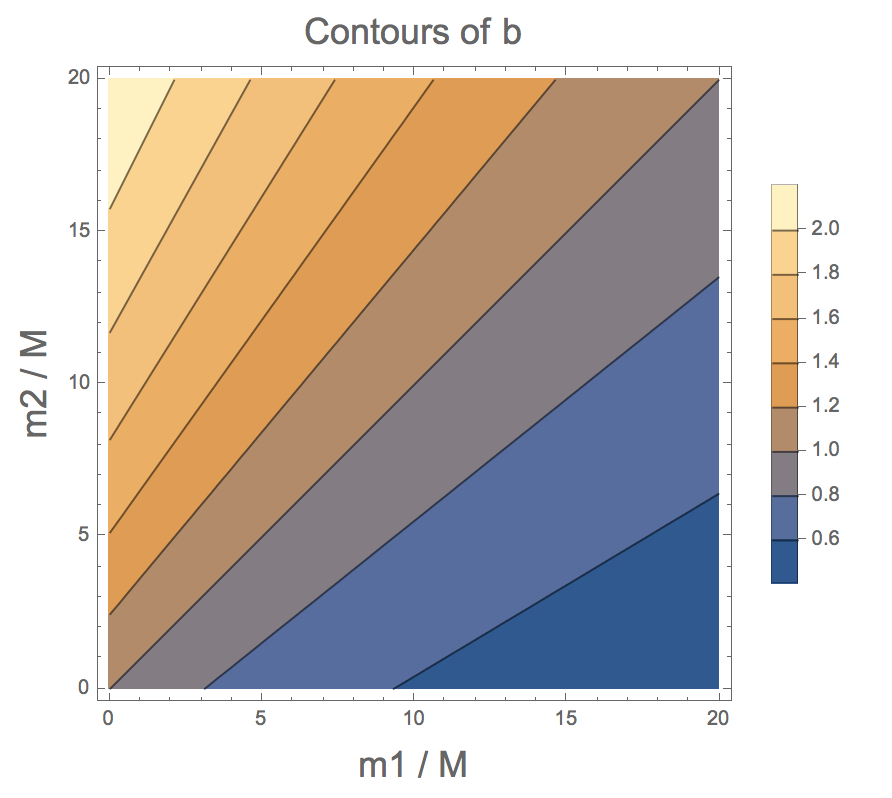}
\caption{Contours showing the value of $b$ as a function of the input mass ratios $m_1/M$ and $m_2/M$.}\label{fig3a}
\end{figure}

\begin{figure}[ht]%
\centering
\includegraphics[width=1.0\textwidth]{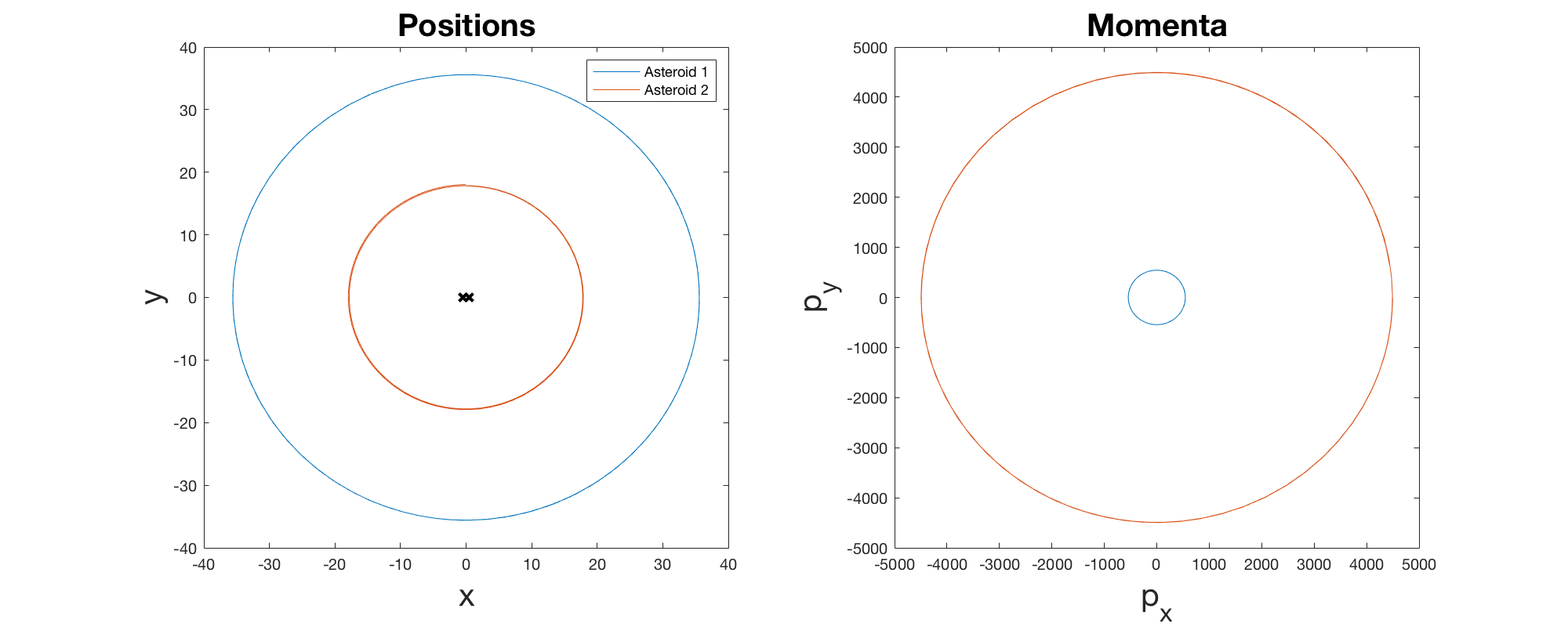}
\caption{A solution of \eqref{eqn34} with $b = 2$ and $\delta = 1$. To obtain this $b$, we had to use extremely large (and thus unrealistic) asteroid masses.}\label{fig3}
\end{figure}

\subsection{Periodic Solution to $O(\delta^3)$}
Once $b$ is determined, we can solve for $r_2$ in the $p_{r_2}$ ODE:
\begin{align*}
0 = \frac{(c_2/m_2)^2}{r_2} - M - \frac{m_1}{(1+b)^2} 
\implies r_2 = \frac{(1+b)^2(c_2/m_2)^2}{M(1+b)^2 + m_1}.
\end{align*}
Some algebraic manipulation involving the use of $q(b) = 0$ reveals
\begin{equation*}
    r_1 = br_2 = \frac{(1+b)^2(c_1/m_1)^2}{M(1+b)^2 + m_2b^2}.
\end{equation*}
We may also obtain $\dot{\theta}_i$:
\begin{equation}
\dot{\theta}_i = \omega = -\frac{2\pi}{P} + \frac{\delta^3[M(1+b)^2 + m_1]^2}{(1+b)^4(c_2/m_2)^3} = -\frac{2\pi}{P} + \frac{\delta^3[M(1+b)^2 + m_2b^2]^2}{(1+b)^4(c_1/m_1)^3}.
\end{equation}
Hence we have obtained two periodic circular solutions to \eqref{eqn29}:
\begin{align}
r_1(t) & = br_2 = \frac{(1+b)^2(c_1/m_1)^2}{M(1+b)^2 + m_2b^2}, & r_2(t) & = \frac{(1+b)^2(c_2/m_2)^2}{M(1+b)^2 + m_1},\nonumber\\
\theta_1(t) & = \pi + \theta_0 + \omega t, & \theta_2(t) & = \theta_0 + \omega t, \nonumber\\
p_{r_1}(t) & = 0, & p_{r_2}(t) & = 0, \nonumber\\
p_{\theta_1}(t) & = c_1 = \frac{m_1}{m_2}b^2c_2, & p_{\theta_2}(t) & = c_2.
\label{eqn34}
\end{align}
The period is
\[ T = \frac{2\pi}{\vert\omega\vert} = \frac{ 2\pi P (1+b)^4(c_2/m_2)^3}{2\pi(1+b)^4(c_2/m_2)^3 -2\pi P \delta^3[M(1+b)^2+m_1]^2}.\]

\subsection{Multipliers and Continuation}
Since there are more degrees of freedom here in the comet case than the Hill case, we shall work with the full $8\times8$ matrix variational equation for the monodromy matrix $\Xi(T)$. The variational equation is $\dot{\Xi} = S\Xi, \Xi(0) = I$, where $S = J_8 D^2(H_0 + \delta^3 H_3)$ along the solution \eqref{eqn34}. We found via Mathematica and verified by hand that $S$ has the following block structure:
\[S = \begin{bmatrix}
-A^T & B \\ C & A
\end{bmatrix},\] 
where
\begin{align*}
A = \begin{bmatrix}
0 & A_1/b & 0 & 0 \\
0 & 0 & 0 & 0 \\
0 & 0 & 0 & A_1 \\
0 & 0 & 0 & 0
\end{bmatrix}, & & B = \begin{bmatrix}
B_1 & 0 & 0 & 0 \\
0 & B_2 & 0 & 0 \\
0 & 0 & B_3 & 0 \\
0 & 0 & 0 & B_4
\end{bmatrix}, & & C = \begin{bmatrix}
C_{11} & 0 & C_{13} & 0 \\
0 & C_{22} & 0 & -C_{22} \\
C_{13} & 0 & C_{33} & 0 \\
0 & -C_{22} & 0 & C_{22} 
\end{bmatrix},
\end{align*}
for the time-independent entries given by
\begin{equation}\label{entriesofS}
\begin{aligned}
A_1 & = \frac{2\delta^3(c_2/m_2)}{r_2^3}, \\
B_1 & = \frac{\delta^3}{m_1},\ B_2 = \frac{\delta^3}{m_1b^2r_2^2},\ B_3 = \frac{\delta^3}{m_2},\ B_4 = \frac{\delta^3}{m_2r_2^2}, \\
C_{11} & = -\frac{\delta^3m_1[M(1+b)^3 + m_2b^2(3+b)]}{b^3(1+b)^3r_2^3},\ C_{13} = \frac{2\delta^3m_1m_2}{(1+b)^3r_2^3},\\
C_{33} & = -\frac{\delta^3m_2[M(1+b)^3 + m_1(1+3b)]}{(1+b)^3r_2^3},\ C_{22} = \frac{\delta^3m_1m_2b}{(1+b)^3r_2}.
\end{aligned}
\end{equation}

Mathematica gave the characteristic polynomial of $S$ as being of the form
\begin{equation}\label{eqn36}
\det(S - zI) = D_1z^2 + D_2z^4 + D_3z^6 + z^8.
\end{equation}
We are interested in $D_1$ only, but for completeness the coefficients are:
\begin{equation}\label{eqnDs}
\begin{aligned}
D_1 & = B_1B_3C_{22}\left[-\frac{A_1^2}{b^2}\left(b^2C_{11} + 2bC_{13} + C_{33}\right) + (B_2 + B_4)\left(C_{13}^2 - C_{11}C_{33}\right)\right], \\
D_2 & = C_{22}\left[\frac{A_1^2}{b_2}\left(B_1 + b^2B_3\right) + (B_2 + B_4)(B_1C_{11} + B_3C_{33})\right]\\
& \ \ \ \ - B_1B_3\left(C_{13}^2 - C_{11}C_{33}\right), \\
D_3 & = -\left[B_1C_{11} + (B_2 + B_4)C_{22} + B_3C_{33}\right].
\end{aligned}
\end{equation}
Since $S$ is time-independent, the monodromy matrix is $\Xi(T) = \exp(ST)$, like it is in the Hill case, and the multipliers will be $e^{zT}$, where $z$ are the roots of \eqref{eqn36}. So, to show that $+1$ is a multiplier exactly twice, we must show that $D_1$ is nonzero. Since $B_1, B_3$, and $C_{22}$ are each positive, it is sufficient to examine $D_1/(B_1B_3C_{22})$.

We write $D_1/(B_1B_3C_{22})$ in terms of $\delta, M, b, r_2$, and the asteroid mass ratios $\mu_i = m_i/M$, $i=1,2$. Using the expressions for $A_1$, $B_2$, $B_4$, $C_{11}$, $C_{13}$, and $C_{33}$ in \eqref{entriesofS}, replacing one of the $r_2$ from the $r_2^9$ in the common denominator of $-(A_1^2/b^2)(b^2C_{11}+2bC_{13}+C_{33})$ with its constant value in the circular periodic solution in \eqref{eqn34}, and using the quintic equation $q(b) = 0$ in \eqref{eqn31} to express each of $b^5$, $b^6$, and $b^7$ as linear combinations of $1$, $b$, $b^2$, $b^3$, and $b^4$, we employed Mathematica, and Maple, and also verified by hand (a lengthy computation), so show that
\[
\frac{D_1}{B_1B_3C_{22}} = \frac{\delta^9M^3}{b^5(1+b)^4r_2^8}\left[k_0 + k_1b + k_2b^2 + k_3b^3 + k_4b^4\right],
\]
where
\begin{align}
k_0 & = 6\mu_1 + 3\mu_2 + \mu_1\mu_2,\nonumber\\
k_1 & = 15\mu_1 + 12\mu_2 + 3\mu_1\mu_2,\nonumber\\
k_2 & = 15\mu_1 + 18\mu_2 + 11\mu_1\mu_2 + 3\mu_1^2 + \mu_2^2 - \mu_1\mu_2^2,\nonumber\\
k_3 & = 9\mu_1 + 12\mu_2 + 8\mu_1\mu_2 - 2\mu_1^2 + 4\mu_2^2 + 2\mu_1^2\mu_2 + 3\mu_1\mu_2^2,\nonumber\\
k_4 & = 3\mu_1 + 3\mu_2 + 2\mu_1\mu_2 - 2\mu_1^2 + 3\mu_2^2 - 2\mu_1^2\mu_2.
\label{eqn37}
\end{align}
The quantities $k_0$ and $k_1$ in \eqref{eqn37} are positive for positive $\mu_1$ and $\mu_2$. There are values of $\mu_1>1,\mu_2>1$ for which each of the quantities $k_2, k_3$, and $k_4$ in \eqref{eqn37} is negative, but for realistic masses (e.g., for $\mu_1,\mu_2 \in (0,1]$), we have
\begin{align*}
k_2 & \geq 15\mu_1 + 18\mu_2 + 10\mu_1\mu_2 + 3\mu_1^2 + \mu_2^2 > 0,\\
k_3 & \geq 7\mu_1 + 12\mu_2 + 8\mu_1\mu_2 + 4\mu_2^2 + 2\mu_1^2\mu_2 + 3\mu_1\mu_2^2 > 0,\\
k_4 & \geq \mu_1 + 3\mu_2 + 3\mu_2^2 > 0.
\end{align*}
Hence, $D_1 > 0$ for all $\mu_1,\mu_2 \in (0,1]$.

All of the nonzero entries of $S$ are a multiple of $\delta^3$, as follows from \eqref{entriesofS}. The matrix $\hat S = \delta^{-3} S$ does not depend on $\delta$ and the eigenvalues of $\hat S$ are determined by roots of
\[  \hat D_1 z^2 + \hat D_2 z^4 + \hat D_3 z^6 + z^8 = z^2(\hat D_1 + \hat D_2 z^2 + \hat D_3 z^4 + z^6),\]
where coefficients $\hat D_i$, $i=1,2,3$ are expressions, independent of $\delta$, given by the formulas in \eqref{eqnDs} in which each of the scalars $A_1$, $B_1$, $B_2$, $B_3$, $B_4$, $C_{11}$, $C_{13}$, $C_{22}$, and $C_{33}$ in \eqref{entriesofS} are each multiplied by $\delta^{-3}$ before being substituted. The quantities $D_1$ and $\hat D_1$ are related by $\delta^9 \hat D_1 = D_1$, so that $D_1>0$ implies $\hat D_1>0$. Hence $\hat S$ has eigenvalue $0$ with algebraic multiplicity two. For two circular periodic solutions in \eqref{eqn34}, one with a positive and one with a negative choice of $c_2$, say $c_2 = \pm 1$ (where $c_1 = \pm b^2(m_1/m_2)$ as determined by \eqref{eqn30}), an adaptation of the argument used in Meyer and Offin \cite{meyerOffin} can now be applied to prove continuation for all small positive values of $\delta$.

\begin{theorem}[Comet Orbits]\label{thmComets}
There exist two one-parameter families of near-circular comet-type periodic solutions in the Binary Asteroid Problem \eqref{eqn6} for all realistic masses $(m_1/M, m_2/M \in (0,1])$. These families tend to infinity.
\end{theorem}

\begin{figure}[ht]%
\centering
\includegraphics[width=1.0\textwidth]{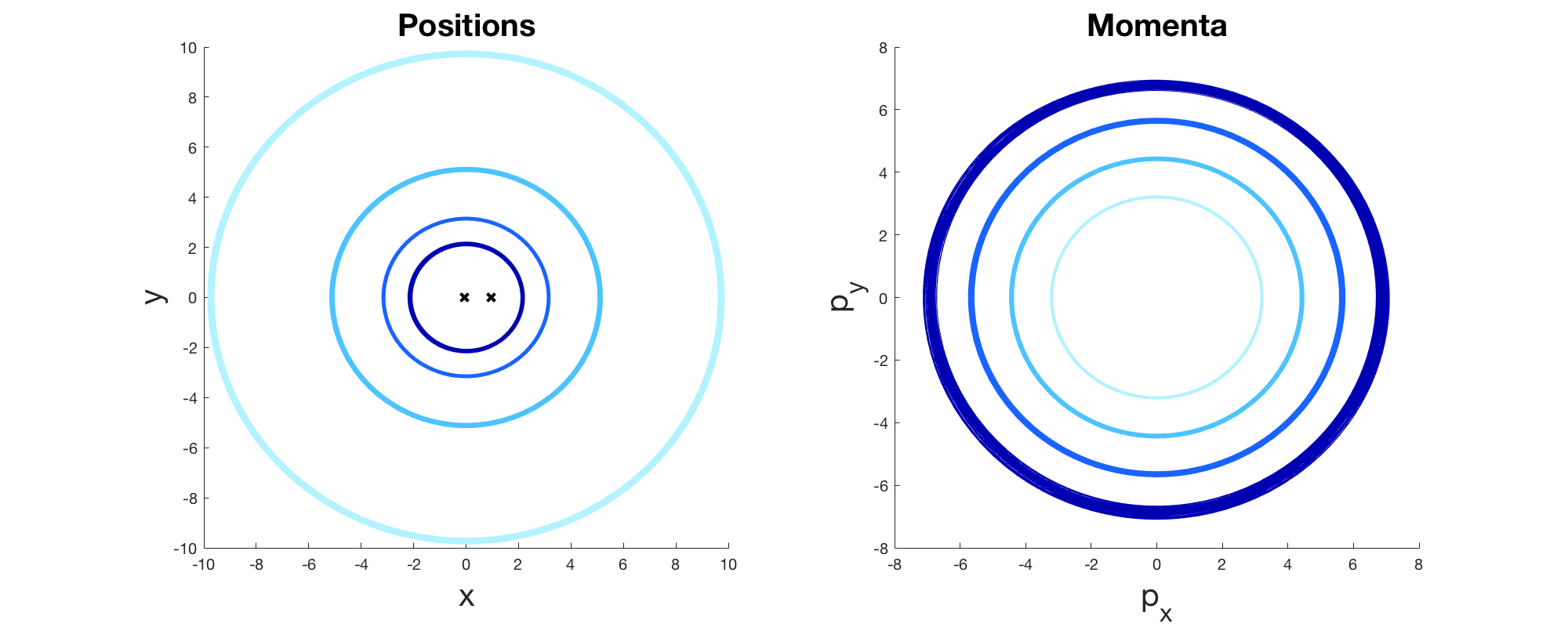}
\caption{Depiction of a numerical $\delta$ continuation process for a comet-type orbit \eqref{eqn34} with $b = 1.0001$, $M_1 = 95$, and $M_2 = 5$.
}\label{fig4}
\end{figure}

\begin{remark} This comet approach can be adapted to develop Hill-type equations of motion similar in form to \eqref{eqn29}. The main difference is that we would employ symplectic shifts to center each asteroid on one primary, then scale. Doing so produces circular Hill solutions out to $\delta^3$. However, multiplier analysis of these solutions shows they have $+1$ as a multiplier four times, and hence they do not continue. We believe that the lack of continuation in this case is due to the coupling term of the potential not being present in the equations of motion until $O(\delta^6)$.
\end{remark}


\section{Conclusion}\label{sec13}
We have presented a planar 4-body model, the Binary Asteroid Problem, with Hamiltonian $H_{BA}$ in \eqref{eqn6}, and its COM reduced Binary Asteroid Problem with Hamiltonian $H_{xy}$ in \eqref{eqn10}. In the COM reduced Binary Asteroid Problem we have the existence of six relative equilibria (Theorem \ref{ExistenceCritPts}), four of which are saddle-centers (Theorem \ref{thmSpectralStability}) and collinear with the primaries, and two of which are hyperbolic (Theorem \ref{thmSpectralStability}) and form convex kites (and each is a rhombus) with the primaries. Near each saddle-center relative equilibrium there is a one-parameter family of periodic orbits (Corollary \ref{LyapunovFamily}). Each relative equilibrium limits to its corresponding relative equilibrium in the Generalized Copenhagen Problem as the common mass of the asteroids goes to $0$ (Theorem \ref{ExistenceCritPts}). In the COM reduced Binary Asteroid Problem we have two one-parameter families of near-circular periodic orbits near the origin (Theorem \ref{Thm:CentralHills}) on which the two asteroids form a binary asteroid. Also in the COM reduced Binary Asteroid Problem we found two one-parameter families of near-circular Hill-type orbits (Theorem \ref{thm2}) where one asteroids orbits one primary and the other asteroid orbits the other primary. In the Binary Asteroid Problem we found two families of near circular periodic comet-type orbits where the two asteroids are approximately opposite each other through the origin (Theorem \ref{thmComets}).

The approaches used in this paper can be adapted to produce results in the spatial Binary Asteroid Problem. We have numerical simulations showing spatial continuation of the Hill-type and comet-type orbits. However, multiplier analysis is much more difficult because the variational equation for each of these spatial orbits becomes time-dependent. In addition to questions (1), (2), (3), (4), and part of (5) posed in the Introduction, future research will also include the study of these Hill-type and comet-type orbits in the spatial Binary Asteroid Problem, and the continuation of planar and spatial periodic orbits in the Binary Asteroid Problem into the four-body problem.




\bibliographystyle{amsplain}
\bibliography{references}

\end{document}